%% file: arxiv_version.tex
\documentclass{article}

\usepackage{microtype}
\usepackage{graphicx}
\usepackage{subcaption}
\usepackage{booktabs} %

\usepackage{hyperref}

\usepackage[accepted]{icml2024}

\usepackage{mathrsfs}
\usepackage{enumitem}
\usepackage{amsthm}
\usepackage[super]{nth}
\usepackage[toc,page]{appendix}
\usepackage{listings}
\usepackage{mathtools}
\usepackage{color}
\usepackage{tabularx}
\usepackage{amsmath}
\usepackage{csquotes}
\usepackage{multicol}
\usepackage{amssymb}
\usepackage{fancyvrb}
\usepackage{makecell}
\usepackage{dsfont}

\usepackage{mdframed}
\usepackage[toc,page]{appendix}
\usepackage{mathdots}
\usepackage{xspace}
\usepackage{empheq}
\usepackage{arydshln}
\usepackage{nicefrac}
\usepackage{dsfont}
\usepackage{thm-restate}
\allowdisplaybreaks
\usepackage{array}
\newcolumntype{L}[1]{>{\raggedright\let\newline\\\arraybackslash\hspace{0pt}}m{#1}}
\newcolumntype{C}[1]{>{\centering\let\newline\\\arraybackslash\hspace{0pt}}m{#1}}
\newcolumntype{R}[1]{>{\raggedleft\let\newline\\\arraybackslash\hspace{0pt}}m{#1}}

\usepackage{tikz}
\usetikzlibrary{automata, positioning, calc, shapes, arrows, fit}

\usepackage{pgfplots}
\DeclareMathOperator{\E}{\mathbb{E}}
\DeclareMathOperator{\A}{\mathcal{A}}
\DeclareMathOperator{\U}{\mathcal{U}}
\DeclareMathOperator{\prob}{Pr}

\DeclareMathOperator{\dist}{dist}
\DeclareMathOperator{\supp}{supp}

\theoremstyle{plain}
\newtheorem{theorem}{Theorem}
\newtheorem{lemma}[theorem]{Lemma}

\newtheorem{example}{Example}

\DeclareMathOperator*{\argmin}{arg\,min}
\DeclareMathOperator*{\soc}{SC}

\usepackage[textsize=tiny]{todonotes}
\usepackage[capitalize,noabbrev]{cleveref}

\usepackage[textsize=tiny]{todonotes}

\icmltitlerunning{Can a Few Decide for Many? The Metric Distortion of Sortition}

\begin{document}

\twocolumn[
\icmltitle{Can a Few Decide for Many? The Metric Distortion of Sortition}

\begin{icmlauthorlist}
\icmlauthor{Ioannis Caragiannis}{iannis}
\icmlauthor{Evi Micha}{evi}
\icmlauthor{Jannik Peters}{jannik}
\end{icmlauthorlist}

\icmlaffiliation{iannis}{Department of Computer Science, Aarhus  University, Aarhus, Denmark}
\icmlaffiliation{evi}{Computer Science, Harvard University, Cambridge, USA }
\icmlaffiliation{jannik}{ Research Group Efficient Algorithms, Faculty IV – Electrical Engineering and Computer Science, TU Berlin, Berlin, Germany}

\icmlcorrespondingauthor{Ioannis Caragiannis}{iannis@cs.au.dk}
\icmlcorrespondingauthor{Evi Micha}{emicha@seas.harvard.edu}
\icmlcorrespondingauthor{Jannik Peters}{jannik.peters@tu-berlin.de}

\icmlkeywords{Machine Learning, ICML}

\vskip 0.3in
]

\printAffiliationsAndNotice{}  %

\begin{abstract}
Recent works have studied  the design of algorithms for selecting representative sortition panels. However, the most central question remains unaddressed: Do these panels reflect the entire population's opinion? We present a positive answer by adopting the concept of metric distortion from computational social choice, which aims to quantify how much a panel's decision aligns with the ideal decision of the population when preferences and agents lie on a metric space. We show that uniform selection needs only logarithmically many agents in terms of the number of alternatives to achieve almost optimal distortion. We also show that Fair Greedy Capture, a selection algorithm introduced recently by~\citet{EbMi23}, matches uniform selection's guarantees of almost optimal distortion and also achieves constant ex-post distortion, ensuring a ``best of both worlds'' performance.
\end{abstract}

\section{Introduction}
In recent years, \emph{sortition} or the randomized selection of so-called \emph{citizens' assemblies} has gained traction around the globe. In such citizens' assemblies, a randomly selected subset of the population discusses and deliberates on societal issues, such as challenges from climate change and AI threats. Notable examples include the citizens' assemblies on climate change held in Belgium, France, or the United Kingdom or the Irish citizens' assembly on abortion which lead to a referendum on the topic and finally the legalization of abortion. This surge in the usage of citizens' assemblies has been supported by work in the AI and ML community, which have focused on the design of the underlying sampling methods \citep{FGG+21a} and the formulation of normative properties of sortition and related analysis \citep{EKM+22a, MST21a, EbMi23}. As a prominent example, the sampling method of \citet{FGG+21a} was recently deployed for the selection of the German citizens' assembly on nutrition.\footnote{See \url{https://www.bundestag.de/en/parliament/citizens_assemblies}.} 

As outlined by the sociologist \citet{Eng89a}, two of the main supposed advantages of sortition are the \emph{representation} of the underlying population and \emph{fairness}. The interpretation of fairness by the aforementioned works is quite straightforward: \emph{a selection procedure is fair if every member of the population is selected for the panel with equal probability}. Fairness can be achieved under {\em uniform selection} where a panel of given size is sampled uniformly at random. 
Representation on the other hand is not that easily defined. Both \citet{EKM+22a} and \citet{EbMi23} assumed that there is an underlying \emph{representation metric}, measuring the closeness between different members of the population, for instance by measuring the similarities of certain attributes. Using this metric, \citet{EKM+22a} measured the distance of each agent to their $q$-th closest member on the panel, and showed that if $q$ is more than half the panel size, uniform selection of the panel achieves a constant factor approximation to summing these distances and thus being in some sense representative of the underlying population. \citet{EbMi23} on the other hand, interpreted representation as \emph{proportional representation} and imported the notion of the core from social choice \citep{gillies1953some}, which intuitively requires that each sufficiently large group has sufficiently many representatives in the panel. They  show that uniform selection satisfies the  core in expectation, and introduce a  fair selection algorithm that always returns a panel in the core.   

However, while previous work has focused on the formation of panels that are representative of the population, it is still quite unexplored whether such panels  reach decisions that reflect the aggregate preference of the underlying population.
So, our goal in the current paper is 
\begin{quote}
  \textit{(a) to explore whether panels chosen using sortition are able to represent the interest of the whole population for different issues, and (b) to provide bounds on the panel size that is sufficient and necessary for achieving this.  }
\end{quote}

\paragraph{Our approach and results.} To measure the extent to which a panel's decision is close to the decision the whole population would collectively take, we utilize the concept of \emph{distortion} from the area of computational social choice \citep{PrRo06a,BCH+15}.
Specifically, first, we extend the definition of the representation metric of~\citet{EKM+22a} and \citet{EbMi23} by assuming that both agents and alternatives are embedded in the same metric space. Then, we compare the (expected) social cost of the best alternative (defined as its total distance from all agents) for the panel to the social cost of the best alternative for the overall population. If the ratio between the two is low, then the will of the panel is close to the will of the electorate and the preferences of the underlying population are well represented. 

In \Cref{sec:general-k}, we start by presenting the benefits of fair selection for achieving low distortion. In scenarios with arbitrarily many alternatives, deterministic selection cannot have a distortion considerably better than $5$ with small panels, while, in contrast, fair selection algorithms always achieve a distortion of $3$. Uniform selection, i.e., selecting each panel of a given size uniformly at random, is the most natural fair selection algorithm. For uniform selection, we specifically ask the question: for parameters $m$ and $\varepsilon$, how small the panel should be so that the selected alternative has an expected social cost that is at most a factor of $1+\varepsilon$ far from the optimal social cost for all instances with $m$ alternatives? In other words, we seek for low {\em ex-ante} distortion. In Section~\ref{sec:uniform}, we prove an upper bound of $\mathcal{O}\left(\varepsilon^{-2}\ln{\frac{m}{\varepsilon}}\right)$ on the panel size. This can be thought of as a ``dimensionality reduction'' result. No matter how large the population is, the panel size depends only on the number of alternatives and the accuracy parameter. Our proof uses a concentration inequality due to \citet{S74} for sums of random variables with limited correlation. We also prove a matching lower bound for every selection algorithm. Our main tool here is a new anti-concentration inequality for the hypergeometric probability distribution.

One of the criticisms of sortition outlined by \citet{Eng89a} is that, in extreme cases, the results will be socially disastrous. Indeed, uniform selection selects the worst possible panel with positive probability. In our terminology, the {\em ex-post} distortion can be as bad as $\Omega(n)$, where $n$ is the population size. In Section~\ref{sec:FGC}, we revisit the Fair Greedy Capture algorithm, introduced recently by \citet{EbMi23}. Even though it was originally designed with a different objective in mind (i.e., forming panels that are representative of a population under the definition of the core), we show that it achieves a ``best of both worlds'' performance guarantee. First, panels of size $\mathcal{O}\left(\varepsilon^{-3}\ln{\frac{m}{\varepsilon}}\right)$ yield an ex-ante distortion of $1+\varepsilon$ for all instances with $m$ alternatives. Second, we show that the ex-post distortion is only constant. The proof of our bound on the panel size exploits structural properties of Fair Greedy Capture that allow for the use of Hoeffding's inequality.

Lastly, in~\Cref{sec:experiments}, we empirically compare uniform selection and Fair Greedy Capture over real data. Among other interesting observations, we notice that Fair Greedy Capture consistently achieves better distortion than uniform selection, which converges quickly to $1$.

\section{Related Work}
We highlight three streams of research closely related to us.

\paragraph{Sortition.} The first one consists of computer science work on formalizing and analyzing sortition. Here, starting with the work of \citet{BGP19a}, there has been a focus on devising sampling methods for sortition \citep{FGG+21a, FGP21a, FGGP20a}. Besides the aforementioned papers by \citet{EKM+22a} and \citet{EbMi23}, the work closest to ours is by
\citet{MST21a}. They also analyze the distortion of sortition. However, they focus on a series of binary issues (i.e., two issues with utility either $0$ or $1$ for the agents) deriving almost tight bounds for this special case, showing that for such simple options the distortion of uniformly selecting a panel of size $k$ is $1 + \mathcal{O}\left(k^{-1/2}\right)$. Our work can be seen as a generalization of theirs, in the sense that we allow many alternatives and general metric utilities, showing that sortition maintains a very low distortion
even in a more general setting. Finally, \citet{AFP22a} studied (among other things) the metric distortion of first forming a panel of agents and then running a voting rule. They were able to show that selecting a panel of size $\mathcal{O}(\varepsilon^{-2}\ln{m})$ uniformly at random and then running Copeland's rule is enough to achieve distortion $5 + \varepsilon$ with high probability, while uniform selection with size $\mathcal{O}(\frac{m}{\varepsilon^2})$ yields distortion $3 + \varepsilon$, using the distortion-$3$ rule of \citet{GHS20a}. Our work additionally improves upon theirs and shows that to achieve distortion $3 + \varepsilon$ (or even better if we use the recent randomized voting rule of~\citet{CRWW24a}), a panel size of $\mathcal{O}\left(\varepsilon^{-2}\ln{\frac{m}{\varepsilon}}\right)$ is also sufficient. 

\paragraph{Distortion.}
Secondly, the field of distortion in computational social choice is apparently relevant. In particular, in recent years, the study of \emph{metric distortion}, in which the preferences of the agents come from a metric space, has gained traction, with several voting rules being shown to have good distortion. For instance, random dictatorship \cite{AnPo16a}, Copeland \cite{ABP15a}, a variant of the vote-by-veto rule \cite{KiKe22a}, or maximal lotteries \cite{CRWW24a} achieve (close to) optimal metric distortion. We refer to \citet{AFSV21a} for a survey and to \citet{ABFV22a}, \citet{GKP+23a}, \citet{EFLN23a}, and \citet{CF23} for some recent works on (non-metric) distortion. The work on distortion most related to ours is most likely the one by \citet{JaSk22a}. In their model, just like in the work of \citet{MST21a}, they measure distortion with regard to a series of binary decisions. However, their model differs in two aspects: (i) the voter set and the set of candidates for the committee are different, and (ii) the voters just provide a ranking over the candidates for the committee. In this setting, \citet{JaSk22a} show that no (deterministic) committee selection method can achieve a bounded distortion, even when restricted to alternatives which are top-choices of some committee.

\paragraph{Clustering.}
Finally, through our usage of the Fair Greedy Capture algorithm, we are inherently related to the proportional clustering literature and the Greedy Capture mechanism introduced by \citet{CFL+19a}. Since Fair Greedy Capture achieves constant factor approximations to both proportional fairness as studied by \citet{CFL+19a} or \citet{LLS+21a} and individual fairness as studied by \citet{JKL20a} (see \citet{KePe23a} for a proof of this and how these notions are related), this also shows that representation in the sense of selecting a panel with good distortion is also compatible with representation in the sense of proportionally or individually fair clustering as devised in the aforementioned works. 

\section{Model and Motivation}
For  $t \in \mathbb{N} $, define $[t]=\{1,\ldots, t\}$. 
Throughout the paper, we assume that we are given a set of agents $N = [n]$ and a set of alternatives $C = \{c_1, \dots, c_m\}$. 
Together $N \cup C$ lie in an underlying  (pseudo-) metric space, with a distance function $d \colon (N \cup C) \times (N \cup C) \to \mathbb{R}_{\ge 0}$. We denote with $d(i,j)$ the distance between  each $i,j\in N \cup C $. We assume that the following properties are satisfied by $d$: (a) $d(i,j)=d(j,i)$ (symmetry) and (b) for $i,j,\ell\in N \cup C$,  $d(i,j)\leq d(i,\ell)+d(\ell,j)$ (triangle inequality).

The \emph{cost} an agent $i \in N$ has for an alternative $c \in C$ is equal to $d(i,c)$, with a lower cost being ``better'' for the agent.  The  social cost of an alternative $c$ is equal to    $\soc(c) \coloneqq \sum_{i \in N} d(i,c)$, i.e., the sum of the costs over all agents. A panel $P$ is a subset of the agents, i.e., $P\subseteq N$. For such a panel $P$ and alternative $c \in C$, we let $\soc(P,c) \coloneqq \sum_{i \in P} d(i,c)$ denote that social cost of that alternative for the panel. Using this, we define $c(P) \coloneqq \argmin_{c \in C} \soc(P,c)$ to be the alternative with the lowest social cost for this panel, with ties among possible alternatives being broken arbitrarily. 

 While the locations of alternatives in the metric space may be unknown, the position of the agents is typically known. For example, each agent is characterized by a set of features crucial for a specific application, and its position in the metric space is a function of these features. Therefore, a {\em selection algorithm} $\A_k$ takes  as input only the locations of the agents  and outputs a distribution over panels of size $k$. A panel is in the support of $\A_k$ if it is returned with non-zero probability. 
 We call a selection algorithm fair, if the probability of any agent being on the panel is equal to $\frac{k}{n}$. We pay special attention to  \emph{uniform selection}, $\U_k$,  which selects each panel of size $k$ with equal probability. Uniform selection is fair by definition.
 
We consider the decision of panels selected by an algorithm to be ``good''  if its social cost is not much larger than the social cost of the optimum $c(N)$.  More formally, we define the {distortion} of a selection algorithm $\A_k$ as the worst-case ratio of the expected social cost of the alternatives selected by panels in the support of $\A_k$ and the minimum possible social cost over any instance, i.e.,
\begin{align*}
    \dist(\A_k)= \sup_d \frac{\E_{P \sim \A_{k}} [SC(c(P)) ] }{\soc(c(N))}.
\end{align*}
We will call this as ex-ante distortion, as well.

\begin{example}\label{ex:example-1}
    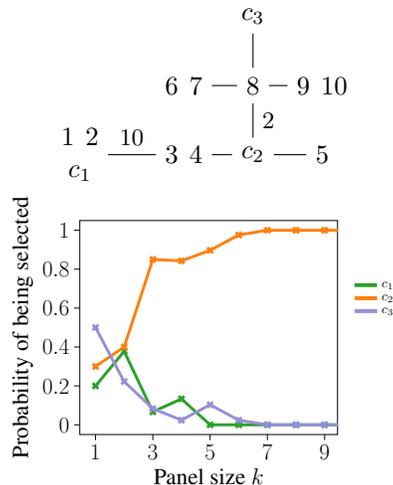
\begin{figure}
        \centering
        \begin{tikzpicture}[scale=0.92]
            \begin{scope}[yscale=-1]
                \node [align=center] (v1) at (.5,0.5) {$1\,\ 2$\\$c_1$};
                
                \node (v4) at (2,0.5) {$3 \,\ 4$};
                
                \node (v5) at (3,0.5) {$c_2$};
                
                \node (v6) at (4,0.5) {$5$};

                \node (v7) at (2,-0.5) {$6 \,\ 7$};

                \node (v8) at (3,-0.5) {$8$};

                \node (v10) at (3,-1.5) {$c_3$};

                \node (v9) at (4,-0.5) {$9 \,\ 10$};
                
                \draw (v1) -- node[above] {\footnotesize $10$} (v4);
                \draw (v4) -- (v5);
                \draw (v5) -- (v6);
                \draw (v5) -- node[right] {\footnotesize $2$} (v8);
                \draw (v7) -- (v8);
                \draw (v8) -- (v9);
                \draw (v8) -- (v10);
            \end{scope}
        \end{tikzpicture}
        \begin{tikzpicture}[scale=0.5, every plot/.append style={line width=2.2pt}]

\definecolor{chocolate2267451}{RGB}{226,74,51}
\definecolor{dimgray85}{RGB}{85,85,85}
\definecolor{gainsboro229}{RGB}{229,229,229}
\definecolor{mediumpurple152142213}{RGB}{152,142,213}
\definecolor{steelblue52138189}{RGB}{52,138,189}
\definecolor{darkorange25512714}{RGB}{255,127,14}
\definecolor{forestgreen4416044}{RGB}{44,160,44}
\begin{axis}[
legend cell align={left},
legend style={
  fill opacity=0.8,
  draw opacity=1,
  draw=none,
  at={(1.25,0.75)},
  text opacity=1,
  line width=3pt
},
legend entries={$c_1$, $c_2$, $c_3$},
tick align=outside,
tick pos=left,
ylabel near ticks,
xlabel near ticks,
x grid style={darkgray176},
xlabel = {Panel size $k$},
xmin=0.45, xmax=9.45,
xtick style={color=black},
y grid style={darkgray176},
ylabel = {Probability of being selected},
ymin=-0.05, ymax=1.05,
ytick style={color=black},every tick label/.append style={font=\LARGE}, 
label style={font=\LARGE},
xticklabels={$1$,$3$,$5$,$7$,$9$},
xtick={1,3,5,7,9}
]
\addlegendimage{forestgreen4416044}
\addlegendimage{darkorange25512714}
\addlegendimage{mediumpurple152142213}
\addplot [thick, forestgreen4416044, mark = x, mark size = 3, mark options = {solid} ]
table {%
1 0.2
2 0.377777777777778
3 0.0666666666666667
4 0.133333333333333
5 0
6 0
7 0
8 0
9 0
10 0
};
\addplot [thick, darkorange25512714, mark = x, mark size = 3, mark options = {solid}]
table {%
1 0.3
2 0.4
3 0.85
4 0.842857142857143
5 0.896825396825397
6 0.976190476190476
7 1
8 1
9 1
10 1
};
\addplot [thick, mediumpurple152142213, mark = x, mark size = 3, mark options = {solid}]
table {%
1 0.5
2 0.222222222222222
3 0.0833333333333333
4 0.0238095238095238
5 0.103174603174603
6 0.0238095238095238
7 0
8 0
9 0
10 0
};
\end{axis}

\end{tikzpicture}
        \caption{Metric spaces for \Cref{ex:example-1}. Edges without labels have length $1$. The graph indicates the probability of the different alternatives being the top choice of a uniformly selected panel, for different panel sizes. }
        \label{fig:example-distortion}

    \end{figure}

    As an example consider the graph depicted in the upper part of \Cref{fig:example-distortion} together with its shortest distance metric, e.g, the distance from agent $1$ to alternative $c_2$ is $10 + 1 = 11$.

    Now for the social costs, one can calculate that the total distance to alternative $c_1$ is $101$ to alternative $c_2$ is $39$ and to alternative $c_3$ is $49$. If a mechanism would select $\{1,2\}$ as the panel, their best alternative would be $c_1$ leading to a distortion of $\frac{101}{39} \sim 2.59$, while for instance $\{5,8\}$ would have $c_2$ as their best alternative with a distortion of $1$. In the lower part of \Cref{fig:example-distortion} we plot the probability for each alternative of being $c(P)$ for a uniformly at random selected panel $P$ of size $k$, with tie-breaking in favor of $c_1$. For instance, for $k = 2$, the probability for alternative $c_1$ is $\frac{17}{45}$, for $c_2$ is $\frac{2}{5}$ , and for $c_3$ it is $\frac{2}{9}$. Thus the distortion of uniform selection for $k = 2$ in this instance is 
    \[
    \frac{\frac{17}{45} 101 + \frac{2}{5} 39 + \frac{2}{9} 49}{39} = \frac{64.6\overline{4}}{39} \sim 1.66. 
    \]
\end{example}

\section{The Benefit of Fair Selection}\label{sec:general-k}
Before we turn to randomized mechanisms, we start the investigation of distortion with a general lower bound on the distortion of {\em any} deterministic mechanism. %
We show that for any panel size $k$ there is an instance in which no deterministic mechanism can reach a distortion of better than $5$. 

\begin{restatable}{theorem}{detlower} \label{det_lower}
Let $k,n$ be positive integers with $k\leq n$ and $\mathcal{A}_k$ a deterministic algorithm that selects a panel of size $k$ among $n$ agents. For any $\varepsilon>0$, there exists an instance with $n$ agents and two alternatives in which $\mathcal{A}_k$ has distortion of at least $5 - \frac{12k}{n+2k}-\varepsilon$.
\end{restatable}

\begin{proof}
Consider an instance with two alternatives $a$ and $b$ and a set $N$ of agents at distance $2$ from each other. Alternative $a$ has distance $3-\varepsilon$ from each agent in $\mathcal{A}_k$ and distance $5-\varepsilon$ from each agent in set $N\setminus \mathcal{A}_k$. Alternative $b$ has distance $3$ from each agent in set $\mathcal{A}_k$ and distance $1$ from each agent in set $N\setminus \mathcal{A}_k$. Finally, we can set $d(a,b) = 4 - \varepsilon$. It can be easily seen that the distances among agents and alternatives form a metric.
Clearly, the panel $\mathcal{A}_k$ will select alternative $a$, who has social cost $\soc(a)=(5-\varepsilon) (n-k)+(3-\varepsilon) k = (5-\varepsilon) n-2k$. The social cost of alternative $b$ is $\soc(b)=n-k+3k=n+2k$. Thus, the distortion is $\frac{(5-\varepsilon)n-2k}{n+2k}\geq 5-\frac{12k}{n+2k}-\varepsilon$.
\end{proof}

So, deterministic algorithms that select panels of size much smaller than the agent population cannot have a distortion considerably better than $5$. In contrast, we can show that \emph{any} fair selection algorithm achieve a strictly better ex-ante distortion.

\begin{restatable}{theorem}{generalupper}
Let $k,n$ be positive integers with $k\leq n$ and $\mathcal{A}_k$ a fair algorithm that selects an assembly of size $k$ among $n$ agents. Then, the ex-ante distortion of $\mathcal{A}_k$ is at most $3-\frac{2k}{n}$. \label{thm:general_dist_metric}
\end{restatable}
\begin{proof}
Consider an alternative $c'$; we will show that $\E[\soc(c(\mathcal{A}_k))] \leq \left(3-2k/n\right)\cdot \soc(c')$. 
For any $k$-sized set $P$ of agents, we clearly have 
\begin{align}\label{eq:upper-3-minus-2k-over-n-1}
    \soc(c(P)) &=\soc(P,c(P))+\soc(N\setminus P,c(P)).
\end{align}
By the definition of alternative $c(P)$, we have
\begin{align}\label{eq:upper-3-minus-2k-over-n-2}
    \soc(P,c(P)) &\leq \soc(P,c').
\end{align}
Also, using the triangle inequality, 
\begin{align}\nonumber
&\soc(N\setminus P, c(P))=\sum_{i\in N\setminus P}{d(i,c(P))}\\\nonumber
&\leq \sum_{i\in N\setminus P}{\left(d(i,c')+d(c',c(P))\right)}\\\label{eq:upper-3-minus-2k-over-n-3}
&=\soc(N\setminus P,c')+(n-k)\cdot d(c',c(P)).
\end{align}
The triangle inequality implies that $d(c',c(P))\leq d(i,c')+d(i,c(P))$ for every alternative $i\in P$. Thus,
\begin{align}\nonumber
    &d(c',c(P)) \leq \frac{1}{k}\sum_{i\in P}{\left(d(i,c')+d(i,c(P))\right)}\\\label{eq:upper-3-minus-2k-over-n-4}
    &=\frac{1}{k}\soc(P,c')+\frac{1}{k}\soc(P,c(P))\leq \frac{2}{k}\soc(P,c'),
\end{align}
where the last inequality follows by equation (\ref{eq:upper-3-minus-2k-over-n-2}). Using equation (\ref{eq:upper-3-minus-2k-over-n-4}), equation (\ref{eq:upper-3-minus-2k-over-n-3}) yields
\begin{align}\label{eq:upper-3-minus-2k-over-n-5}
    \soc(N\setminus P,c(P)) &\leq \soc(N\setminus P,c')+2\frac{n-k}{k}\soc(P,c').
\end{align}
Now, using equations (\ref{eq:upper-3-minus-2k-over-n-2}) and (\ref{eq:upper-3-minus-2k-over-n-5}), equation (\ref{eq:upper-3-minus-2k-over-n-1}) yields
\begin{align}\nonumber
&\soc(c(P))\\\nonumber
&\leq \soc(P,c')+\soc(N\setminus P,c')+2\frac{n-k}{k}\soc(P,c')\\\label{eq:upper-3-minus-2k-over-n-6}
&=\soc(c')+2\frac{n-k}{k}\soc(P,c').
\end{align}
Thus, using equation (\ref{eq:upper-3-minus-2k-over-n-6}), we have 
\begin{align}\nonumber
    &\E[\soc(c(\mathcal{A}_k))] = \sum_{P\in N_k}{\Pr[\mathcal{A}_k=P]\cdot \soc(c(P))}\\\nonumber
    &\leq \sum_{P\in N_k}{\Pr[\mathcal{A}_k=P]\cdot \soc(c')}\\\label{eq:upper-3-minus-2k-over-n-7}
    &\quad\quad +2\frac{n-k}{k}\sum_{P\in N_k}{\Pr[\mathcal{A}_k=P]\cdot \soc(P,c')}.
\end{align}
The first sum in the RHS of equation (\ref{eq:upper-3-minus-2k-over-n-7}) becomes
\begin{align}\nonumber
&\sum_{P\in N_k}{\Pr[\mathcal{A}_k=P]\cdot \soc(c')}\\\label{eq:upper-3-minus-2k-over-n-8}
&= \soc(c')\cdot \sum_{P\in N_k}{\Pr[\mathcal{A}_k]}=\soc(c'),
\end{align}
while, since algorithm $\mathcal{A}_k$ is fair, the second one becomes
\begin{align}\nonumber
&\sum_{P\in N_k}{\Pr[\mathcal{A}_k=P]\soc(P,c')}\\\nonumber
&= \sum_{P\in N_k}{\Pr[\mathcal{A}_k=P]\cdot \sum_{i\in P}{d(i,c')}}\\\nonumber
&= \sum_{i\in N}{d(i,c')\cdot \sum_{P\in N_k:i\in P}{\Pr[\mathcal{A}_k=P]}}\\\nonumber
&=\sum_{i\in N}{d(i,c')\cdot \Pr[i\in \mathcal{A}_k]}\\\label{eq:upper-3-minus-2k-over-n-9}
&=\frac{k}{n}\cdot \sum_{i\in N}{d(i,c')}=\frac{k}{n}\cdot \soc(c').
\end{align}
Hence, using equations (\ref{eq:upper-3-minus-2k-over-n-8}) and (\ref{eq:upper-3-minus-2k-over-n-9}), equation (\ref{eq:upper-3-minus-2k-over-n-7}) yields
\begin{align*}
    \E[\soc(c(\mathcal{A}_k))] &\leq \left(3-\frac{2k}{n}\right)\cdot \soc(c'),
\end{align*}
as desired.
\end{proof}
We note that for the case of $k = 1$ this is the bound for the famous random dictatorship mechanism \citep{AnPo16a}. Comparing the exact bound derived in \Cref{det_lower}, this shows that to even theoretically being able to match the ex-ante distortion of fair mechanisms for $k = 1$, a deterministic mechanism would need to have at least a panel size of $k = \frac{n}{4}$. 

We can further show that the bound in Theorem~\ref{thm:general_dist_metric} is tight.

\begin{restatable}{theorem}{loweruniform}
    For every $\varepsilon>0$ and positive integers $k,n$ with $k\leq n$, there exists an instance in which any selection algorithm $\mathcal{A}_k$ has an ex ante distortion of at least $3-\frac{2k}{n}-\varepsilon$.
    \label{thm:general_lower_uniform}
\end{restatable}
\begin{proof}
Consider the instance with a set $N$ of $n$ agents and $n \choose k$+1 alternatives defined as follows. There is an alternative $c$ at distance $3+\varepsilon$ from every agent in $N$. For every distinct set $K$ of $k$ agents from $N$, there exists an alternative $c_K$ at distance $3$ from the agents in $K$ and distance $9$ from the agents in $N\setminus K$. Notice that every panel $K$ returned by algorithm $\mathcal{A}_k$ with positive probability will select alternative $c_K$, who has social welfare $\soc(c_K)=3k+9(n-k)=9n-6k$. Alternative $c$ has minimum social welfare of $\soc(c)=(3+\varepsilon)n$. The ex ante distortion of algorithm $\mathcal{A}_k$ is then
$\frac{9n-6k}{(1+\varepsilon)n} = 3-\frac{2k}{n}-\frac{\varepsilon(3n-2k)}{(3+\varepsilon)n}\geq 3-\frac{2k}{n}-\varepsilon$.
\end{proof}
Note that the instance constructed in the proof of~\Cref{thm:general_lower_uniform} requires an exponential number of alternatives in terms of the number of agents. However, in the next sections, we see that when the number of alternatives does not depend on the number of agents, we can achieve almost optimal distortion, by sampling only logarithmically many number of agents in terms on the number of alternatives $m$. 

\section{Uniform Selection}\label{sec:uniform}

We now consider uniform selection and show that panels of size $\mathcal{O}(\varepsilon^{-2}\ln{\frac{m}{\varepsilon}})$ yield ex-ante distortion of $1+\varepsilon$ for every instance with $m$ alternatives.\footnote{Throughout this paper we assume that $\varepsilon \to 0$ for the $\mathcal{O}$-notation. } To prove the next theorem, we need to bound the contribution to the expected social cost of the alternative that can be returned by the panel, distinguishing between three classes of alternatives. Alternatives that can incur (ex-post) distortion at most $1+\varepsilon/2$ are easy to handle. To take care of alternatives of distortion between $1+\varepsilon/2$ and $73$, we use Serfling's inequality for the concentration of sampling without replacement \citep{S74}. However, applying this to alternatives incurring distortion higher than $73$ would introduce a (logarithmic) dependency on the number of agents, since the distortion of such alternatives can be $\Omega(n)$. So, to bound the contribution of these alternatives, we construct an explicit hypergeometric random variable underlying the selection process and use it instead to bound the selection probability appropriately.

\begin{restatable}{theorem}{maintheorem}
For every $\varepsilon\in (0,1]$ and integer $m>0$, the exists integer $k\in \mathcal{O}\left(\varepsilon^{-2}\ln{\frac{m}{\varepsilon}}\right)$, so that a uniformly random $k$-sized panel of agents, has ex ante distortion at most $1+\varepsilon$, for every set of $m$ alternatives.
    \label{main_theorem}
\end{restatable}

\begin{proof}
Let $\varepsilon\in (0,1]$ and $k$ be the smallest integer multiple of $3$ that is at least $\max\left\{\frac{25}{2\varepsilon^2}\ln{\frac{144m}{\varepsilon}},3+3\log{\frac{72m}{\varepsilon}}\right\}$. Clearly, $k\in \mathcal{O}\left(\varepsilon^{-2}\ln{\frac{m}{\varepsilon}}\right)$. \footnote{Here $\ln$ is the logarithm with base $e$ and $\log$ is the logarithm with base $2$. }
We will abbreviate $c(N)$ as $c'$.  We denote by $W$ the set of alternatives of social cost at most $\left(1+\frac{\varepsilon}{2}\right)\cdot \soc(c')$. We also denote by $P$ the random panel selected. Unless otherwise specified, expectations and probabilities are defined over the selection of set $P$ among all $k$-sized subsets of agents uniformly at random. 
Then,
\begin{align}\nonumber
    &\E[\soc(c(P))]\\\nonumber 
    &\leq \left(1+\frac{\varepsilon}{2}\right)\cdot \Pr[c(P)\in W]\cdot \soc(c')\\\nonumber
    &\quad\quad +\sum_{a\in C\setminus W}{\Pr[c(P)=a]\cdot \soc(a)}\\\nonumber
    &\leq \left(1+\frac{\varepsilon}{2}\right)\cdot \left(1-\sum_{a\in C\setminus W}{\Pr[c(P)=a]}\right)\soc(c')\\\nonumber
    &\quad\quad +\sum_{a\in C\setminus W}{\Pr[c(P)=a]\cdot \soc(a)}\\\nonumber
    &\leq \left(1+\frac{\varepsilon}{2}\right)\cdot \soc(c')\\\label{eq:break-distortion-bound}
    &\quad\quad +\sum_{a\in C\setminus W}{\Pr[c(P)=a]\cdot (\soc(a)-\soc(c'))}.
\end{align}
We will show that, for every alternative $a\in C\setminus W$, it holds that
\begin{align}\label{eq:bound-prob-for-candidate-not-in-W}
    \Pr[c(P)=a] &\leq \frac{\varepsilon\cdot \soc(c')}{2m\cdot (\soc(a)-\soc(c'))}.
\end{align}
This, using (\ref{eq:bound-prob-for-candidate-not-in-W}), equation (\ref{eq:break-distortion-bound}) will yield 
\begin{align*}
    \E[\soc(c(P))]
    &\leq \left(1+\frac{\varepsilon}{2}\right)\cdot \soc(c')+ \frac{\varepsilon|C\setminus W|}{2m}\cdot \soc(c')\\
    &\leq (1+\varepsilon)\cdot \soc(c'),
\end{align*}
as desired.

In the following, we prove inequality (\ref{eq:bound-prob-for-candidate-not-in-W}) for every alternative $a\in C\setminus W$. We distinguish between alternatives in $C\setminus W$ of low and high social cost.

\textbf{Case 1: $\soc(a)-\soc(c')\leq  72\cdot \soc(c')$.} In this case we use a concentration inequality by~\citet{S74}.
\begin{theorem}[\citet{S74}]\label{thm:serfling74}
Let $c_1$, $c_2, \dots, c_n$ be $n$ real numbers from the interval $[\alpha,\beta]$ and consider the process of selecting $k$ of them uniformly at random without replacement. Let $X$ be the random variable denoting the sum of the $k$ numbers selected. Then, $\E[X]=\frac{k}{n}\sum_{i=1}^n{c_i}$ and for every $t>0$, its holds that
\begin{align*}
\Pr[X-\E[X] \geq t] &\leq \exp\left(-\frac{2t^2}{k(\beta-\alpha)^2}\right).
\end{align*}
\end{theorem}

Recall that for any panel $P\subseteq N$ of agents,
\begin{align*}
\soc(P,c')-\soc(P,a) &=\sum_{i\in P}{(d(i,c')-d(i,a))}.
\end{align*}
Furthermore, using $D=d(a,c')$, by the triangle inequality, we have
\begin{align*}
    -D &\leq d(i,c')-d(i,a)\leq D,
\end{align*}
for every agent $i\in N$. Thus, when selecting panel $P$ uniformly at random among all $k$-sized sets of agents, the random variable $\soc(P,c')-\soc(P,a)$ follows the distribution of random variable $X$ in Theorem~\ref{thm:serfling74} with $c_i=d(i,c')-d(i,a)$ for $i\in N$, $\alpha=-D$, and $\beta=D$. Then, Theorem~\ref{thm:serfling74} yields
\begin{align*}
    \E[\soc(P,c')- \soc(P,a)]&=-\frac{k}{n}\left(\soc(a)-\soc(c')\right)
\end{align*}
and, using $t=\frac{k}{n}\cdot (\soc(a)-\soc(c'))$,
\begin{align}\nonumber
    &\Pr[\soc(P,c')\geq \soc(P,a)]\\\label{eq:prob-good-candidate-a-wins-opt}
    &\leq \exp\left(-\frac{2k(\soc(a)-\soc(c'))^2}{n^2D^2}\right).
\end{align}
Notice that,%
\begin{align}\nonumber
    \soc(a)-\soc(c') &\geq \frac{\varepsilon}{5}(\soc(a)+\soc(c'))\\\nonumber
    &=\frac{\varepsilon}{5}\cdot \sum_{i\in N}{(d(i,a)+d(i,c'))}\\\label{eq:soc-difference-good-candidate}
    &\geq \frac{\varepsilon n D}{5}.
\end{align}
The first inequality follows by our assumption that $a\in C\setminus W$ and since $\varepsilon\in (0,1]$ and the second one by the triangle inequality. 

Now, notice that an alternative $a$ is selected by the panel $P$ only if $\soc(P,c')\geq \soc(P,a)$. Using this observation, together with equations (\ref{eq:soc-difference-good-candidate}) and (\ref{eq:prob-good-candidate-a-wins-opt}), and the facts that $k\geq \frac{25}{2\varepsilon^2} \ln{\frac{144m}{\varepsilon}}$ and $\soc(a)-\soc(c')\leq 72 \cdot\soc(c')$, we obtain
\begin{align*}
    \Pr[c(P)=a] &\leq \Pr[\soc(P,c')\geq \soc(P,a)]\\
&\leq \exp\left(-2k\varepsilon^2/25\right)\leq \frac{\varepsilon}{144m}\\
&\leq \frac{\varepsilon\cdot \soc(c')}{2m\cdot (\soc(a)-\soc(c'))},
\end{align*}
completing the proof of inequality (\ref{eq:bound-prob-for-candidate-not-in-W}) in Case 1.

\textbf{Case 2: $\soc(a)-\soc(c')\geq 72\cdot \soc(c')$.} Define $L=\left\{i\in N: d(i,c')>\frac{D}{4}\right\}$ and let $\ell = |L|$. We will prove two useful lemmas.

\begin{restatable}{lemma}{smallell}\label{lem:small-ell}

    It holds that \(\ell \leq \frac{4n\cdot \soc(c')}{\soc(a)-\soc(c')}.\)
\end{restatable}
\begin{proof}
By the definition of set $L$, we have $\soc(c')\geq \sum_{i\in L}{d(i,c')}>\frac{D\ell}{4}$. Thus, we have
    \begin{align*}
        \soc(a)-\soc(c') &= \sum_{i\in N}{(d(i,a)-d(i,c'))}\\
        &\leq nD\leq \frac{4n\cdot \soc(c')}{\ell}.
    \end{align*}
\end{proof}

\begin{restatable}{lemma}{badcand}\label{lem:bad-candidate-wins}
    For every $k$-sized set $P$ of agents, it that holds $c(P)=a$ only if $|L\cap P|\geq k/3$.
\end{restatable}
\begin{proof}
Assume otherwise that $|L\cap P|<k/3$. Using the triangle inequality and the definition of set $L$, it holds that $d(i,a)-d(i,c')\geq D-2d(i,c')\geq D/2$ for $i\not\in L$. Also, again by the triangle inequality, $d(i,a)-d(i,c')\geq -D$ for $i\in L$. Thus,
\begin{align*}
    \soc(P,a)-\soc(P,c')&=\sum_{i\in P\cap L}{(d(i,a)-d(i,c'))}\\
    &\quad\quad +
    \sum_{i\in P\setminus L}{(d(i,a)-d(i,c'))}\\
    &\geq |P\setminus L|\cdot D/2-|P\cap L|\cdot D\\
    &=kD/2-|P\cap L|\cdot 3D/2>0,
\end{align*}
contradicting the assumption that $c(P)=a$.
\end{proof}

Lemma~\ref{lem:bad-candidate-wins} implies that $\Pr[c(P)=a]=0$ if $\ell<k/3$. In the following, we bound the probability $\Pr[c(P)=a]$ assuming that $\ell\geq k/3$. 

Notice that the quantity $|P\cap L|$ follows a hypergeometric distribution with a population size of $n$, $\ell$ success states, and $k$ draws. The mode of the  distribution of $|P\cap L|$ is at most $\frac{(k+1)(\ell+1)}{n+2}$ which, due to Lemma~\ref{lem:small-ell} and the assumption $\soc(a)-\soc(c')\geq 72\cdot \soc(c')$ in Case 2, is at most $\frac{(k+1)\cdot (n/18+1)}{n+2}<k/3$. Thus, $\Pr[|P\cap L|=j]$ is decreasing in $j$ for $j\geq k/3$ and
\begin{align}\nonumber
    &\Pr[|P\cap L|\geq k/3]=\sum_{j=k/3}^k{\Pr[|P\cap L|=j]}\\\nonumber
    &\leq \left(\frac{2k}{3}+1\right)\Pr[|P\cap L|=k/3]\\\nonumber
    &=\left(\frac{2k}{3}+1\right) \frac{{\ell \choose k/3} \cdot {n-\ell \choose 2k/3}}{{n \choose k}}\\\nonumber
    &= \left(\frac{2k}{3}+1\right) {k \choose k/3} \cdot \prod_{i = 1}^{k/3} \frac{(\ell - k/3 + i)}{n - k/3 + i}\\\nonumber
    &\quad\quad \cdot\prod_{i = 1}^{2k/3} \frac{(n-\ell - 2k/3 + i)}{n - k + i}\\\label{eq:prob-bad-candidate-wins}
    &\leq \left(\frac{2k}{3}+1\right) {k \choose k/3} \left(\frac{\ell}{n}\right)^{k/3}\leq \left(\frac{9\ell}{n}\right)^{k/3}.
\end{align}
The second inequality follows since $k/3\leq \ell\leq n$ and the third one by verifying that $(\frac{2k}{3}+1)\cdot {k \choose k/3} \leq 9^{k/3}$ for every $k\geq 3$ which is a multiple of $3$. 

Now, we use Lemma~\ref{lem:bad-candidate-wins}, inequality (\ref{eq:prob-bad-candidate-wins}), the fact that $k\geq 3+3\log{\frac{72m}{\varepsilon}}$, Lemma~\ref{lem:small-ell}, and the assumption $\soc(a)-\soc(c')\geq 72\cdot\soc(c')$ in Case 2 to get
\begin{align*}
    \Pr[c(P)=a] &\leq \Pr[|P\cap L|\geq k/3]\\
    &\leq \left(\frac{9\ell}{n}\right)^{1+\log{\frac{72m}{\varepsilon}}}\\
    &\leq \left(\frac{36\cdot \soc(c')}{\soc(a)-\soc(c')}\right)^{1+\log{\frac{72m}{\varepsilon}}}\\
    &\leq \frac{36\cdot \soc(c')}{\soc(a)-\soc(c')}\cdot \left(\frac{1}{2}\right)^{\log{\frac{72m}{\varepsilon}}}\\
    &=\frac{\varepsilon\cdot \soc(c')}{2m\cdot (\soc(a)-\soc(c')}.
\end{align*}
This completes the proof of inequality (\ref{eq:bound-prob-for-candidate-not-in-W}) in Case 2 and, subsequently, of Theorem~\ref{main_theorem}.
\end{proof}

We now show that Theorem~\ref{main_theorem} is nearly tight.

\begin{restatable}{theorem}{lowerlogm}\label{thm:lower-bound-small-panels}
    Let $\varepsilon\in (0,1/30]$, and $n,m,k$ be positive integers such that $n$ is even, $m\geq 1+6^9$, and $k\leq \min\left\{\frac{n}{2},m-1, \frac{\ln{(m-1)}}{288\varepsilon^{2}}\right\}$. Then, there exists a set of $n$ agents such that every probability distribution over panels of agents of size $k$ has distortion more than $1+\varepsilon$ for some instance with $m$ alternatives.
\end{restatable}

\section{Distortion Guarantees for Fair Greedy Capture}\label{sec:FGC}

One downside of uniform selection is that it does not ensure that the decision of every realized  panel is good for the whole population, as we see in the following example. 

\begin{example}\label{example:ex-post-distortion}
    Consider an instance with two alternatives, $a$ and $b$, such that  $k$ agents have distance equal to $0$ from  $a$ and distance equal to $1$ from  $b$, while the remaining $n-k$ agents have distance equal to $0$ from  $b$ and distance equal to $1$ from  $a$. With non-zero probability, under uniform selection we can see a panel $P$ that consists from the first $k$ agents, and then $c(P)=a$. Then,    $\soc(c(P))/\soc(c(N))\geq (n-k)/k$. As $k$ decreases this ratio goes to $n$. 
\end{example}

Under this observation, in this section,  we define the {\em ex-post} distortion of a selection algorithm $\A_k$ as the  worst-case ratio of the largest social cost of the alternatives selected by panels in the support of $\A_k$ and the minimum possible social cost, i.e.,
\begin{align*}
    \sup_d\sup_{P \in \supp(\A_{k})} \frac{ \soc(c(P))  }{\soc(c(N))}. 
\end{align*} 
The above example indicates that uniform selection  achieves  bad ex-post distortion. This problem arises because uniform selection does not guarantee that any realized panel is a good representative of the whole population. The key question here is whether we can guarantee that {\em each} selected panel achieves good ex-post distortion, without affecting the expected distortion.

We provide an affirmative answer to this question by revisiting Fair Greedy Capture (Algorithm \ref{alg:fgc}) introduced by \citet{EbMi23}. %
This algorithm  works as follows: It starts  with an empty panel $P$ and continuously grows a ball around every agent with rate $\delta$, until there is a ball that contains at least $\lceil \frac{n}{k}\rceil$ agents.\footnote{This can be equivalently described by calculating the distance of each agent from her $\lceil n/k \rceil$ closest agent, and pick the ball that is centered at the agent with the smallest such distance  and has a radius equal to this distance.}. Then,  one agent within this ball (if there are more than $\lceil n/k \rceil$ agents,  we arbitrarily exclude agents at the boundaries of the ball)  is selected uniformly at random and added to the panel and the rest of the $\lceil \frac{n}{k} \rceil$ agents are temporarily deleted. This process continues this way, until  less than $\lceil \frac{n}{k} \rceil$ agents are left. Then, at most one agent is selected among the agents that have not be disregarded  yet by picking each with probability $\lceil n/k \rceil$, and the last positions in the panel are filled by picking each agent that has not been selected so far with the same probability.\footnote{This description is slightly different than the one in the original paper. We can see that under this version again each agent is selected with probability equal to $k/n$, since before the last step of the algorithm each agent is selected with probability exactly equal to $1/\lceil n/k\rceil$, and at the last step the remaining probability is distributed equally among the agents that have not been selected yet.} %

\begin{algorithm}[tb]
   \caption{Fair Greedy Capture \citep{EbMi23}}
   \label{alg:fgc}
\begin{algorithmic}
   \STATE {\bfseries Input:}  $(N, d)$,  $k$
   \STATE {\bfseries Output:} Panel $P$
   \STATE $R \gets N, \delta \gets 0, P \gets \emptyset$;
   \WHILE{$\lvert R \rvert \ge \lceil \frac{n}{k} \rceil$}
   \STATE Smoothly increase $\delta$;
   \WHILE{there exists $i \in R$ such that $\lvert B(i, \delta) \rvert \ge \lceil \frac{n}{k} \rceil$}
   \STATE  $S \leftarrow \lceil n/k \rceil$ agents from $B(i, \delta)$ by breaking ties arbitrarily  among agents at the boundaries of the ball;
   \STATE Choose $j$ from $S$ uniformly at random;
   \STATE $P \gets  P \cup \{j\}$;
   \STATE $R \gets R\setminus S$;
   \ENDWHILE
   \ENDWHILE
   \IF{$\lvert P \rvert < k$}
   \STATE Sample at most one agent from $R$ by picking each agent in $R$ with probability equal to $1/\lceil n/k \rceil$;
   \STATE Sample the last $k - \lvert P \rvert$ members from $N\setminus P$ uniformly at random;
   \ENDIF
\end{algorithmic}
\end{algorithm}
Here, we show that any panel in the support of this algorithm achieves constant distortion when $k\geq 3$.
\begin{restatable}{theorem}{expostfgc}\label{thm:ex-post-FGC}
    For any $k \ge 3$, Fair Greedy Capture  guarantees ex-post distortion of at most $127$.
\end{restatable}
\begin{proof}
Let $P$ be any panel of size $k$ that has positive probability to be chosen by Fair Greedy Capture. Let $D= d(c(P),c(N))$ and denote by $L$ the set of agents that have distance less than $D/12$ from $c(N)$, i.e., $L:=\{i\in N\colon d(i,c(N)< D/12)\}$. 

Note that
\begin{align}
    &\soc(c(P))=\sum_{i\in N } d(i,c(P)) \nonumber\\
  &  \leq  \sum_{i\in N\setminus L} (d(i,c(N))+D) +\sum_{i\in L} (d(i,c(N))+D) \nonumber \\  
&\leq \sum_{i\in N\setminus L} 13\cdot d(i,c(N)) + \sum_{i\in L}d(i,c(N)) + \lvert L\rvert \cdot D \nonumber  \\ 
&\leq  13\cdot SC(c(N))+\lvert L\rvert \cdot D. \label{ineq:ex-post-1}
\end{align}  
The first inequality follows from the triangle inequality. The second one is due to the fact that for each $i\in N\setminus L$, it holds that $d(i,c(N)) \geq D/12$.  From the same fact, we also get that 
\begin{align}
    \soc(c(N))\geq\sum_{i\in N\setminus L} d(i,c(N))\geq (n-|L|)\cdot D/12.  \label{ineq:ex-post-2}
\end{align}
Therefore, from \Cref{ineq:ex-post-1} and \Cref{ineq:ex-post-2}, we have 
\begin{align}
    \soc(c(P))\leq  \left(13+12\frac{|L|}{n-|L|}\right)\cdot \soc(c(N)). \label{ineq:ex-post-3}
\end{align}
Next, we show that  $|L|$ cannot be significantly large when $c(P)\neq c(N)$.

Suppose that $\lvert L\rvert$ satisfies $\ell\cdot \lceil n/k\rceil \leq \ell < (\ell+1) \cdot \lceil n/k\rceil)$ for some integer $\ell$ with $0\leq \ell\leq k$.  Note that since $|L|\geq \ell \cdot  \lceil n/k\rceil$, we can partition $L$ into at least $\ell$ disjoint parts such that each of them contains at least $\lceil n/k\rceil$ agents with distance less than $D/6$ from each other (since each $i\in L$ has distance less than $D/12$ from $c(N)$).  This means that during the execution of Fair Greedy Capture, at least $\ell$ balls captured individuals from $L$ and each such ball has diameter less than $D/6$.  Therefore, there are at least $\ell$ representatives in $P$ that have distance less than $D/4$ from $c(N)$ (since each such representative has distance less than $D/6$ from an agent in $L$ and each agent in $L$ has distance less than $D/12$ from $c(N)$). 
Thus, for the subset $P'\coloneqq\{i\in P\colon d(i,c(N))< D/4\}$ it holds that $|P'|\geq \ell$ (and therefore $\lvert P\setminus P'\rvert \le k-\ell$). In other words, each panel contains at least $\ell$ representatives from $L$.

Since $c(P)$ is the alternative selected by panel $P$, we have
\begin{align*}
    0 &\geq \soc(P,c(P))-\soc(P,c(N))\\
    &=\sum_{i\in P'}{(d(i,c(P))-d(i,c(N)))}\\
    &\quad\quad+\sum_{i\in P\setminus P'}{(d(i,c(P))-d(i,c(N)))}\\
    &\geq \sum_{i\in P'}{(D-2d(i,c(N)))}-\sum_{i\in P\setminus P'}{D}\\
    &> \ell\cdot D/2-(k-\ell)\cdot D = (3\ell/2-k)\cdot D.
\end{align*}
The second inequality follows by the triangle inequality using $d(i,c(P))\geq D-d(i,c(N)$ for $i\in P'$ and $d(i,c(P))\geq d(i,c(N))-D$ for $i\in P\setminus P'$. The third inequality follows by the definition of set $P'$. We conclude that $\ell<2k/3$, which can be expressed as $\ell\leq \left\lceil \frac{2k}{3}\right\rceil-1$ since $k$ and $\ell$ are integers.

By the definition of $\ell$, we have that $|L| < (\ell+1)\cdot \lceil n/k\rceil$, implying that
\begin{align}\label{eq:L1}
    |L| &\leq \left\lceil\frac{2k}{3}\right\rceil\cdot \left\lceil \frac{n}{k}\right\rceil-1.
\end{align}
Furthermore, panel has $k-\ell$ agents from set $N\setminus L$. Thus, 
\begin{align}\label{eq:L2}
    |L| &\leq n-k+\ell\leq n-k+\left\lceil \frac{2k}{3}\right\rceil-1.
\end{align}
We will use equations (\ref{eq:L1}) and (\ref{eq:L2}) to show that $|L|\leq 19n/21$; this is done in the next lemma.

\begin{lemma}\label{lem:19-over-21}
\begin{align*}
    \min\left\{\frac{\left\lceil \frac{2k}{3} \right\rceil \cdot \left\lceil \frac{n}{k}\right\rceil-1}{n}, 1-\frac{k-\left\lceil\frac{2k}{3}\right\rceil+1}{n}\right\} \leq 19/21.
\end{align*}    
\end{lemma}

\begin{proof}
Notice that if $n\leq 9k-9\left\lceil \frac{2k}{3}\right\rceil+9$, then the second term in the $\min$-expression above is at most $8/9$ and the claim is clearly true. 

In the following, we assume that $n\geq 9k-9\left\lceil \frac{2k}{3}\right\rceil+9$. Notice that the first term in the $\min$-expression is equal to $\frac{\left\lceil\frac{2k}{3}\right\rceil \cdot t-1}{n}$ for $n=(t-1)k+1, (t-1)k+2, \dots, tk$ and integer $t\geq 1$, and is thus decreasing for values of $n$ in that range. We will show that the value $\frac{\left\lceil\frac{2k}{3}\right\rceil\cdot t-1}{(t-1)k+1}$ of the first term for $n=(t-1)k+1$ is at most $19/21$, for all possible integers $t\geq 1$ and $k\geq 3$.

Actually, the assumption $n\geq 9k-9\left\lceil \frac{2k}{3}\right\rceil+9$ implies that we have to consider values of $t$ satisfying $(t-1)k+1\geq 9k-9\left\lceil \frac{2k}{3}\right\rceil+9$, which yields $t\geq 10-\frac{9}{k}\cdot \left\lceil \frac{2k}{3}\right\rceil+\frac{8}{k}$. We distinguish between four cases:

If $\left\lceil\frac{2k}{3}\right\rceil= \frac{2k}{3}$, then $t\geq 4+\frac{8}{k}$, i.e., $t\geq 5$. Thus, $\frac{\left\lceil \frac{2k}{3}\right\rceil\cdot t-1}{(t-1)k+1}\leq \frac{\frac{2k}{3}\cdot 5-1}{4k+1}<5/6<19/21$.

If $\left\lceil\frac{2k}{3}\right\rceil= \frac{2k+1}{3}$ and $k\geq 7$, then
$t\geq 4+\frac{5}{k}$, i.e., $t\geq 5$. Thus, $\frac{\left\lceil \frac{2k}{3}\right\rceil\cdot t-1}{(t-1)k+1}\leq \frac{\frac{2k+1}{3}\cdot 5-1}{4k+1}<5/6<19/21$.

If $\left\lceil\frac{2k}{3}\right\rceil =\frac{2k+2}{3}$, then $k\geq 5$ and $t\geq 4+\frac{2}{k}$, i.e., $t\geq 5$. Thus, $\frac{\left\lceil \frac{2k}{3}\right\rceil\cdot t-1}{(t-1)k+1}\leq \frac{\frac{2k+2}{3}\cdot 5-1}{4k+1}=\frac{5}{6}+\frac{3}{8k+2}\leq 19/21$.

If $k=4$ then $t\geq 6$ and $\frac{\left\lceil \frac{2k}{3}\right\rceil\cdot t-1}{(t-1)k+1}\leq \frac{\left\lceil\frac{8}{3}\right\rceil\cdot 6-1}{5\cdot 4+1}=17/21<19/21$.
\end{proof}
Using \Cref{lem:19-over-21}, \Cref{ineq:ex-post-3} now yields
\begin{align*}
    \soc(c(P)) &\leq \left(13+12\cdot \frac{19/21}{1-19/21}\right)\cdot \soc(c(N))\\
    &=127\cdot \soc(c(N)),
\end{align*}
as desired.
\end{proof}

Further, we show that for $k \leq 2$ there exists an instance, in which Fair Greedy Capture can have linear ex-post distortion in terms of the size of the population.
\begin{restatable}{proposition}{omegan}
     For $k \leq 2$ there exists an instance in which Fair Greedy Capture has $\Omega(n)$ ex-post distortion. 
\end{restatable}
\begin{proof}
    When $k=1$, Fair Greedy Capture operates as uniform selection, and therefore the statement follows  by~\Cref{example:ex-post-distortion}.
    
    For $k=2$, consider $n$ agents on the line, with one agent located at position $1$ and $n-1$ agents located at position $0$. Fair Greedy Capture could select a panel consisting of one agent on $0$ and the agent on $1$. Now consider an instance with two alternatives, one at location $-\delta$ (for $\delta\in (0,1)$) and one at $1$. 
     The panel would choose the alternative at $1$. However, the ex-post distortion of this alternative is $\frac{n-1}{1+n\cdot \delta}$, which approaches $n-1$ as $\delta$ approaches $0$.
\end{proof}
Now we exploit the fact that the outcomes from Fair Greedy Capture are at most a constant factor away from the social cost of the optimum alternative, together with some structural properties to show that $k = \mathcal{O}\left(\varepsilon^{-3}\ln{\frac{m}{\varepsilon}}\right)$ is sufficient to reach a distortion of $1 + \varepsilon$.

\iftrue
\begin{figure*}[t!]
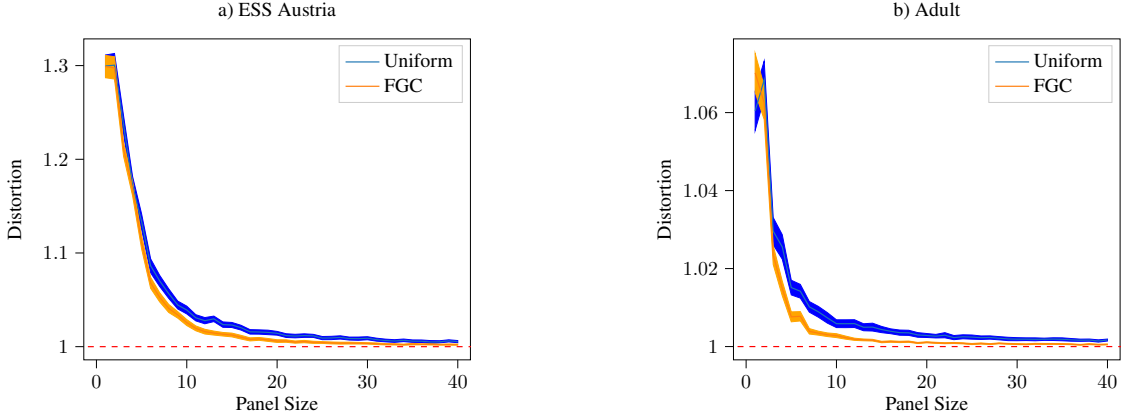

\begin{subfigure}{0.5\textwidth}
\centering
  \include{figures/at}
  \end{subfigure}
  \begin{subfigure}{0.5\textwidth}
  \centering
    \include{figures/adult}
   \end{subfigure}
    \caption{Average distortion for $k$ from $1$ to $40$. The figures include the 95\% confidence interval.}
    \label{fig:exante}
\end{figure*}

\begin{figure*}[t!]
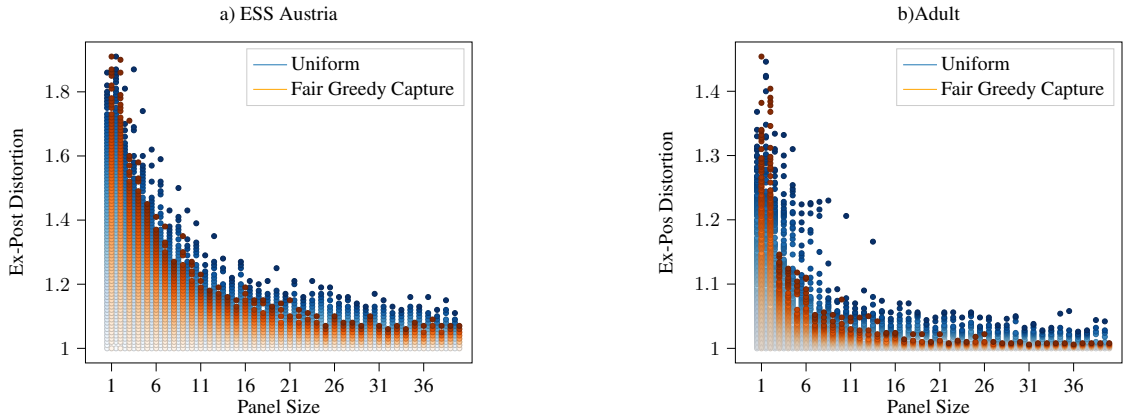

\begin{subfigure}{0.5\textwidth}
\centering
\include{figures/ATscatter}
\end{subfigure}
\begin{subfigure}{0.5\textwidth}
\centering
\include{figures/adultscatter}
\end{subfigure}
    \caption{Distribution of ex-post distortion values for $k$ from $1$ to $49$.}
    \label{fig:expost}
\end{figure*}
\fi

\begin{restatable}{theorem}{exantefgc}\label{thm:ex-ante-FGC}
    For every $\varepsilon>0$ and integer $m>0$, the exists integer $k\in \mathcal{O}\left(\varepsilon^{-3}\ln{\frac{m}{\varepsilon}}\right)$, so that selecting a $k$-sized panel of agents with Fair Greedy Capture, has ex-ante distortion at most $1+\varepsilon$, for every set of $m$ alternatives.
\end{restatable}
Given that both uniform selection and Fair Greedy Capture are fair selection algorithms and that there is a general constant upper bound for fair selection algorithms, one might wonder whether ${\text{poly}}(\varepsilon^{-1})\cdot \ln{\frac{m}{\varepsilon}}$ samples are sufficient for any fair selection algorithm. We give a simple example that this is not the case. 
\begin{restatable}{theorem}{fairimp}
    There exists a fair selection algorithm that has distortion arbitrarily close to $2$ on instances with an even number of agents and two alternatives.
\end{restatable}
\begin{proof}
Let $\delta\in (0,1)$ and $k=n/2$. Consider an instance on the line with $n/2$ agents at location $0$, $n/2$ agents at location $1$, an alternative at location $-1+\delta$ and an alternative at location $1$. Let $\mathcal{A}_k$ be the selection algorithm, which returns a panel consisting, equiprobably, of all agents at location $0$ or all agents at location $1$. The algorithm is clearly fair as each agent is selected with probability $k/n=1/2$. Notice that the panel of agents located at position $0$ will select the alternative at location $-1+\delta$ of social cost $(3/2-\delta)n$, and the panel of agents located at position $1$ will select the alternative at location $1$ of social cost $n/2$. The expected social cost of the alternative selected by $\mathcal{A}_k$ is thus $(1-\delta/2)n$, which yields an ex-ante distortion $2-\delta$, approaching $2$ as $\delta$ approaches $0$.
\end{proof}

\section{Experiments}\label{sec:experiments}

We close off our paper with an experimental section studying the distortion achieved by both uniform selection and Fair Greedy Capture on real world data. 
\subsection{Experimental Setup}
We use the same datasets used by \citet{EbMi23} and \citet{EKM+22a} in their  papers: Firstly, the Adult dataset is probably one of the most popular datasets in research on algorithmic fairness \citep{FMSS22a}. It contains $32,561$ data points from a 1994 US population survey. Following \citet{EbMi23} we use the attributes sex, race, workclass, marital.status, and education.num from this dataset. Secondly, the European Social Survey (ESS) collection of datasets consists of population survey datasets from multiple European countries. We also use the data from \emph{ESS Round 9} from 2018, more specifically, the data from Austria ($n = 2498$).%
\paragraph{Representation Metric Construction.} We follow the method of \citet{EKM+22a}, and use the features of the datasets as a proxy for generating  representation metrics. Similarly, with the aforementioned work, we make the natural assumption that the more characteristics two individuals share the ``closer'' they feel to each other. In more detail, for each dataset, we have a set of features $F = \{f_1, \dots, f_\ell\}$. If a feature is {\em categorical}, e.g., sex, the distance of two individuals with respect to this feature is equal to $1$ if the values  are different and $0$ otherwise. For {\em continuous} features,  the distance of two individuals with respect to this feature is equal to absolute  difference of the  two values, normalized by the maximum distance. We then uniformly at random select a weight between $0$ and $1$ for each feature and set the weighted sum of all feature distances as the actual distances between the agents. 

The original datasets contains only features of individuals, which are  perceived as the agents. To construct the alternatives,  we  use the original agents as a proxy. That is, given a sampled panel $P$, we compare the agent (i.e. alternative) minimizing the distance to the panel $P$ to the agent minimizing the distance to the whole population.

 For each dataset, we average over  $10$ different generated metrics by randomly assigned weights to each feature and over $50$ sample panels from each metric. Average results are plotted with 95\% confidence intervals.  
\subsection{Our Results}
 In \Cref{fig:exante}, we see the average distortion of uniform selection and  Fair Greedy Capture for ESS Austria (left) and Adult (right) datasets. 
 In the ESS Austria dataset, we observe that the distortion of both algorithms converges to $1$ more rapidly than in the Adult dataset. This difference is attributed to the larger number of alternatives ($m=32561$) in the Adult dataset compared to the ESS Austria dataset ($m=2498$). As discussed in~\Cref{sec:uniform} and~\Cref{sec:FGC}, the size of  panel  required to achieve nearly optimal distortion depends on $m$.  %

 Another interesting phenomenon is that while both uniform selection and Fair Greedy Capture behave very similarly for $k = 1$ and $k = 2$, Fair Greedy Capture seems to converge to a distortion of $1$ more quickly than uniform selection and consistently achieves better distortion overall. The former similarity arises because when $k=1$, uniform selection and Fair Greedy Capture perform identically. Moreover, as detailed in~\Cref{sec:FGC}, when $k=2$, Fair Greedy Capture can return panels that exhibit significant distortion. The latter observation is particularly surprising, indicating that despite theoretical bounds suggesting both methods require the same number of samples for convergence to almost optimal distortion, Fair Greedy Capture consistently outperforms uniform selection.

\paragraph{Ex-Post Distortion.}
In~\Cref{fig:expost}, we can also observe the ex-post distortion of the two algorithms over any instance~\footnote{Due to the large number of instances (500 in total), we rounded the values to the second digit to enhance visual clarity} under the same datasets.  As expected, for small values of $k$, both algorithms exhibit outliers, representing realized panels with high distortion. With increasing values of $k$, Fair Greedy Capture consistently achieves low distortion, in contrast to uniform selection, where the presence of outliers remains significant. Notably, in the Adult dataset, uniform selection shows more extreme outliers than in the ESS-Austria dataset. This difference is due to the fact that in the Adult dataset, a significant percentage of the population ($\sim 8.3$\%) is concentrated at the same point, and in some cases, uniform selection fails to choose enough representatives from this group. On the other hand, Fair Greedy Capture always ensures proportional representation of any sufficiently cohesive group.

\section{Discussion}

This work addresses a critical gap in understanding the ability of representative sortition panels to reflect the preferences of an entire population. There are several open questions. First, we show that Fair Greedy Capture achieves ex-post distortion equal to at most $127$, and from~\Cref{det_lower}, we know that no deterministic algorithm can ensure distortion better than $5$. Closing this gap is an immediate open question. Another interesting direction is to show that indeed Fair Greedy Capture always achieves better ex-ante distortion than uniform selection, something that is observed in the experiments.

More broadly, in this work, we assume that the optimal alternative for a population is the one that minimizes the social cost. Various alternatives, such as the one minimizing the maximum cost or the one preferred by the majority in one-to-one comparisons, could be considered. Additionally, in real-life applications, a sortition panel may converge to multiple suggestions and not to just one. The question then arises: Do these suggestions proportionally reflect the opinions of different groups within the underlying population? These intriguing questions are left for future work.

\section*{Acknowledgments}
The work of Ioannis Caragiannis was partially supported by the Independent Research Fund Denmark (DFF) under grant 2032-00185B. Jannik Peters was partially funded by the Deutsche Forschungsgemeinschaft (DFG) under the grant
BR 4744/2-1 and the Graduiertenkolleg “Facets of Complexity” (GRK 2434).

\section*{Impact Statement}

This paper explores two critical questions: (a) whether representative panels selected from a broader population can reach good decisions for the entire population, and (b) the optimal size these panels should be to achieve this goal. While this work is purely theoretical, the findings can have societal implications, particularly if the studied algorithms are implemented in real-world scenarios. Throughout the paper, we discuss the immediate disadvantages of the different approaches that should be taken into account before these algorithms are implemented, despite the positive results we provide in their favor. 

\bibliography{abb, algo, sort}
\bibliographystyle{icml2024}

\newpage
\appendix
\onecolumn

\section{Proof of \Cref{thm:lower-bound-small-panels}}
\lowerlogm*
\begin{proof}
We will define a family of instances with the same set of $n$ agents and different sets of $m$ alternatives and a probability distribution $\mathcal{I}$ to select a random instance among them. All instances have an alternative of social cost $(2-4\varepsilon)n$, and all other alternatives of social cost $2n$. Then, we will show that any $k$-sized panel $P$ satisfies $\E_{\mathcal{I}}[\soc(c(P))]\geq (2-2\varepsilon)n$. Notice that $\E_{\mathcal{I}}[\soc(c(P))]$ is the expectation of the social cost of alternative $c(P)$, taken over the instances produced by $\mathcal{I}$. Thus, for every selection algorithm $\mathcal{A}_k$, the expectation, taken over the random selection of the instance according to $\mathcal{I}$, of the expected social cost of alternative $c(\mathcal{A}_k)$ is at least $(2-2\varepsilon)n$. Then, using the probabilistic method, we have that there exists a single instance for which $\E[\soc(c(\mathcal{A}_k))]\geq (2-2\varepsilon)n$. The ex-ante distortion of $\mathcal{A}_k$, when applied to this instance, will then be $\frac{(2-2\varepsilon)n}{(2-4\varepsilon)n} > 1+\varepsilon$; this will complete the proof.

Our construction of the family of instances is as follows. First, there is a set $N$ of $n$ agents, who have distance $2$ to each other. Let $\mathcal{S}$ be the collection of all subsets of $n/2$ agents from $N$. For every different $(m-1)$-tuple $(S_1, S_2, \dots, S_{m-1})$ of subsets from $\mathcal{S}^{m-1}$, we define the instance with alternatives $c_1$, $c_2, \dots, c_{m-1}$ so that alternative $c_i$ is at distance $1$ from the agents in $S_i$ and at distance $3$ from the agents in $N\setminus S_i$. All instances have an additional alternative $c_0$ at distance $2-4\varepsilon$ from every agent. We can easily verify that the distances between agents and alternatives form a metric. Clearly, in all instances defined in this way, alternative $c_0$ has social cost $(2-4\varepsilon)n$, while all other alternatives have social cost $2n$.

Now, the probability distribution $\mathcal{I}$ returns the instance identified by the $(m-1)$-tuple $(S_1, S_2, \dots, S_{m-1})$
by selecting, for $i=1,2, \dots, m-1$, set $S_i$ uniformly at random among all elements of $\mathcal{S}$, so that the selection is independent for the different $i$s.

We will prove the desired property of the probability distribution $\mathcal{I}$ using the following lemma.

\begin{lemma}\label{lem:prob-half}
Consider a panel $P$ of $k$ agents. Then, $\Pr_{\mathcal{I}}[c(P)=c_0]\leq \frac{1}{2}$.
\end{lemma}
The crucial property our construction and the probability distribution $\mathcal{I}$ have is that, for $i\in [m-1]$, the quantity $|P\cap S_i|$ is a random variable following the hypergeometric distribution with $n$ as population size, $n/2$ success states, and $k$ draws. Furthermore, the random variables $|P\cap S_i|$ and $|P\cap S_j|$ are independent for $i\not=j$.

Let us warm up by proving the lemma assuming that $k\leq 9$.  Using the inequality $(t/s)^s\leq {t \choose s}\leq (et/s)^s$ for positive integers $s$ and $t$ with $t\geq s$, we have that, for $i\in [m-1]$, the probability that all agents in $P$ belong to set $S_i$ is at least 
\begin{align*}
    \frac{{n/2 \choose k}}{{n \choose k}} &\geq \frac{\left(\frac{n}{2k}\right)^k}{\left(\frac{en}{k}\right)^k}=\frac{1}{(2e)^k} \geq \frac{1}{6^9}.
\end{align*}
Notice that we would have $\soc(P,c_i)=k<\soc(P,c_0)$ in this case. Thus, omitting $\mathcal{I}$ from notation and writing just $\Pr$ instead of $\Pr_{\mathcal{I}}$, we can bound the probability that panel $P$ selects alternative $c_0$ as follows:
\begin{align*}
    \Pr[c(P)=c_0] &\leq \left(1-\frac{1}{6^9}\right)^{m-1} \leq 1/e < 1/2.
\end{align*}
The second inequality follows since $m\geq 1+6^9$ and using the inequality $(1-1/t)^t\leq 1/e$ for $t>1$.

In the following, we assume that $k\geq 10$. Observe that alternative $c_0$ is selected by the panel only if $\soc(P,c_i)\geq \soc(P,c_0)$ for every $i\in [m-1]$. Furthermore, observe that $\soc(P,c_i)=|P\cap S_i|+3(k-|P\cap S_i|)=3k-2|P\cap S_i|$ and recall that $\soc(P,c_0)=(2-2\varepsilon)k$. Thus, the condition $\soc(P,c_i)\geq \soc(P,c_0)$ is equivalent to $|P\cap S_i|\leq \frac{k}{2}+2\varepsilon k$. Using these observations, we have
\begin{align}\nonumber
    &\Pr[c(P)=c_0]\\\nonumber
    &\leq \prod_{i=1}^{m-1}{\Pr[\soc(P,c_i)\geq \soc(P,c_0)]}\\\nonumber
    &= \prod_{i=1}^{m-1}{\Pr\left[|P\cap S_i|\leq \frac{k}{2}+2\varepsilon k\right]}\\\label{eq:prob_c0-1}
    &=\prod_{i=1}^{m-1}{\left(1-\Pr\left[|P\cap S_i|> \frac{k}{2}+2\varepsilon k\right]\right)}.
\end{align}
Now, let $\ell^*$ be the smallest integer that is strictly higher than $\frac{k}{2}+2\varepsilon k$. Then, inequality (\ref{eq:prob_c0-1}) implies that
\begin{align}\label{eq:prob_c0-2}
    \Pr[c(P)=c_0] &\leq \prod_{i=1}^{m-1}{\left(1-\Pr[|P\cap S_i|=\ell^*]\right)}.
\end{align}
To continue, we will need an anti-concentration bound for the random variable $|P\cap S_i|$. The proof appears in Appendix~\ref{app:sec:anti-concentration}.

\begin{restatable}{lemma}{hyperconc}\label{lem:anti-concentration}
Let $k$ be an integer such that $10\leq k\leq n/2$. For every integer $\ell$ with $k/2\leq \ell\leq 2k/3$, it holds that
\begin{align*}
    \Pr[|P\cap S_i| = \ell] &\geq \frac{1}{\sqrt{k}}\cdot \left(\frac{k}{\ell}-1\right)^{2\ell-k}
\end{align*}    
\end{restatable}

Notice that $\ell^*<1+k/2+2\varepsilon k$ and, using the assumptions $k\geq 10$ and $\varepsilon\leq 1/30$, we get $\ell^*<k(3/5+2\varepsilon)\leq 2k/3$. By applying Lemma~\ref{lem:anti-concentration} for $\ell=\ell^*$ and substituting the value of $\Pr[|P\cap S_i|=\ell^*]$ in (\ref{eq:prob_c0-2}), we get
\begin{align}\label{eq:prob_c0-final}
\Pr[c(P)=c_0] &\leq \left(1-\frac{1}{\sqrt{k}}\left(\frac{k}{\ell^*}-1\right)^{2\ell^*-k}\right)^{m-1}  
\end{align}

We will complete the proof of Lemma~\ref{lem:prob-half} by distinguishing between two cases for $\ell^*$.

{\bf Case 1.} $\ell^*>\frac{k}{2}+ 3\varepsilon k$. Then, by the definition of $\ell^*$, it is $\ell^*-1<\frac{k}{2}+2\varepsilon k$ which implies that $
\varepsilon k < 1$. Hence, $2\ell^*-k < 2+4\varepsilon k<6$, i.e., $2\ell^*-k\leq 5$ since both $\ell^*$ and $k$ are integers. Recall that $\ell^*\leq 2k/3$. Then, we have $\left(\frac{k}{\ell^*}-1\right)^{2\ell^*-k} \geq (1/2)^{2\ell^*-k}\geq 1/32$. Now, inequality (\ref{eq:prob_c0-final}) yields
\begin{align*}
    \Pr[c(P)=c_0] &\leq \left(1-\frac{1}{32\sqrt{k}}\right)^{m-1} \leq \exp\left(-\frac{m-1}{32\sqrt{k}}\right)\\
    &\leq \exp\left(-\frac{\sqrt{m-1}}{32}\right)\leq 1/e<1/2.
\end{align*}
The second inequality follows since $1-t\leq e^{-t}$, the third one follows by our assumption $k\leq m-1$, and the fourth one is due to the assumption $m\geq 1+6^9$.

{\bf Case 2.} $\ell^*\leq \frac{k}{2}+3\varepsilon k$. Thus,
$\frac{k}{\ell^*}-1 \geq 1-\frac{12\varepsilon k}{k+6\varepsilon k}\geq 1-12\varepsilon$ and inequality (\ref{eq:prob_c0-final}) yields
\begin{align*}
    \Pr[c(P)=c_0] &\leq \left(1-\frac{1}{\sqrt{k}}\left(1-12\varepsilon\right)^{6\varepsilon k}\right)^{m-1}\\
    &\leq \left(1-\frac{1}{\sqrt{k}}\cdot e^{-144\varepsilon^2 k}\right)^{m-1}\\
    &\leq \exp\left(-\frac{m-1}{\sqrt{k}}\cdot e^{-144\varepsilon^2 k}\right)\\
    &\leq \exp\left(-\sqrt{m-1}\cdot e^{-\frac{1}{2}\ln{(m-1)}}\right)\\
    &=1/e<1/2.
\end{align*}
The second inequality follows since $\varepsilon\in (0,1/30]$ and using the property $1-t\geq e^{-2t}$ for $t\in [0,1/2]$. The third inequality follows since $1-t\leq e^{-t}$ and the fourth one is due to our assumption $k\leq \frac{\ln{(m-1)}}{288\varepsilon^2}$.
\end{proof}

We are now ready to complete the proof of Theorem~\ref{thm:lower-bound-small-panels}. Notice that Lemma~\ref{lem:prob-half} implies that, for any panel $P$ of $k$ agents,
\begin{align*}
\E_{\mathcal{I}}[\soc(c(P))] &= (2-4\varepsilon)n\cdot \prob_{\mathcal{I}}[c(P)=c_0]\\
&\quad\quad +2n\cdot \prob_{\mathcal{I}}[c(P)\not=c_0]\\
&=2n-4\varepsilon n\cdot \prob_{\mathcal{I}}[c(P)=c_0]\\
&\geq (2-2\varepsilon)n.
\end{align*}
The proof of the desired property of distribution $\mathcal{I}$ is complete, and the theorem follows.

\section{Proof of Lemma~\ref{lem:anti-concentration}}\label{app:sec:anti-concentration}
\hyperconc*
\begin{proof}
The random variable $|P\cap S_i|$ follows the hypergeometric distribution with a population size of $n$, $n/2$ success states, and $k$ draws. Thus,
\begin{align}\label{eq:factorials}
    \Pr[|P\cap S_i|=\ell] &=\frac{{n/2 \choose \ell}{n/2 \choose k-\ell}}{{n \choose k}} = \frac{(n/2)!^2\cdot k!\cdot (n-k)!}{n!\cdot \ell!\cdot (n/2-\ell)!\cdot (k-\ell)!\cdot (n/2-k+\ell)!}.
\end{align}
We will now use Stirling's approximation which states that $t!=\alpha(t)\cdot \beta(t)\cdot \gamma(t)$ with $\alpha(t)=\sqrt{2\pi t}$, $\beta(t)=\left(\frac{t}{e}\right)^t$, and $\gamma(t)\in \left[e^{\frac{1}{12t+1}}, e^{\frac{1}{12t}}\right]$. Then, equation (\ref{eq:factorials}) can be written as $\Pr[|P\cap S_i|=\ell]=A\cdot B\cdot \Gamma$, where 
\begin{align*}
A &= \frac{\alpha(n/2)^2\cdot \alpha(k)\cdot \alpha(n-k)}{\alpha(n)\cdot \alpha(\ell)\cdot \alpha(n/2-\ell)\cdot \alpha(k-\ell)\cdot \alpha(n/2-k+\ell)},\\
B &= \frac{\beta(n/2)^2\cdot \beta(k)\cdot \beta(n-k)}{\beta(n)\cdot \beta(\ell)\cdot \beta(n/2-\ell)\cdot \beta(k-\ell)\cdot \beta(n/2-k+\ell)}, \mbox{ and}\\
    \Gamma &=\frac{\gamma(n/2)^2\cdot \gamma(k)\cdot \gamma(n-k)}{\gamma(n)\cdot \gamma(\ell)\cdot \gamma(n/2-\ell)\cdot \gamma(k-\ell)\cdot \gamma(n/2-k+\ell)}.
\end{align*}

Using the definition of $\alpha(t)$, we have
\begin{align*}
A &=\sqrt{\frac{kn(n-k)}{2\pi\ell(n-2\ell)(k-\ell)(n-2k+2\ell)}}\geq \sqrt{\frac{k^3}{4\pi \ell^2(k-\ell)^2}} \geq \frac{2}{\sqrt{\pi k}}.
\end{align*}
The first inequality follows by observing that the assumption $k\leq n/2$ implies that $\frac{n}{n-2\ell}\geq \frac{k}{k-\ell}$ and $\frac{n-k}{n-2k+2\ell}\geq \frac{k}{2\ell}$. The second inequality follows since the expression $\ell^2(k-\ell)^2$ is maximized for $\ell=k/2$ at $k^4/16$. 

Using the definition of $\beta(t)$, we have
\begin{align*}
    B &= \frac{\left(\frac{n}{2e}\right)^n\cdot \left(\frac{k}{e}\right)^k\cdot \left(\frac{n-k}{e}\right)^{n-k}}{\left(\frac{n}{e}\right)^n\cdot \left(\frac{\ell}{e}\right)^\ell\cdot \left(\frac{n-2\ell}{2e}\right)^{n/2-\ell}\cdot \left(\frac{k-\ell}{e}\right)^{k-\ell}\cdot \left(\frac{n-2k+2\ell}{2e}\right)^{n/2-k+\ell}}\\
    &=\frac{k^k\cdot (n-k)^{n-k}}{2^k\cdot \ell^\ell\cdot (n-2\ell)^{n/2-\ell}\cdot (k-\ell)^{k-\ell}\cdot (n-2k+2\ell)^{n/2-k+\ell}}\\
    &=\left(\frac{k(n-2k+2\ell)}{2(k-\ell)(n-k)}\right)^k\cdot \left(\frac{(k-\ell)(n-2\ell)}{\ell(n-2k+2\ell)}\right)^\ell\cdot \left(\frac{(n-k)(n-k)}{(n-2\ell)(n-2k+2\ell)}\right)^{n/2}\\
    &\geq \left(\frac{\ell}{k-\ell}\right)^k\cdot \left(\frac{k-\ell}{\ell}\right)^{2\ell}=\left(\frac{k}{\ell}-1\right)^{2\ell-k}.
\end{align*}
The inequality follows since the assumption $k/leq n/2$ implies that $\frac{n-2k+2\ell}{n-k}\geq \frac{2\ell}{k}$ and $\frac{n-2\ell}{n-2k+2\ell}\geq \frac{k-\ell}{\ell}$, while among all pairs of positive numbers with a given sum their product is maximized when they are equal. Thus, $\frac{(n-k)(n-k)}{(n-2\ell)(n-2k+2\ell)}\geq 1$.

Finally, using the definition of $\gamma(t)$, we have \begin{align*}
    \Gamma &\geq \frac{1}{\gamma(n)\cdot \gamma(\ell)\cdot \gamma(n/2-\ell)\cdot \gamma(k-\ell)\cdot \gamma(n/2-k+\ell)}\\
    &\geq \exp\left(-\frac{1}{12n}-\frac{1}{12\ell}-\frac{1}{12(n/2-\ell)}-\frac{1}{12(k-\ell)}-\frac{1}{12(n/2-k+\ell)}\right)\\
    &\geq \exp\left(-\frac{1}{24k}-\frac{1}{6k}-\frac{1}{4k}-\frac{1}{4k}-\frac{1}{6k}\right)=\exp\left(-\frac{7}{8k}\right)\geq \exp(-7/80)>\frac{\sqrt{\pi}}{2}.
\end{align*}
In the third inequality, we have used the inequalities $k\leq n/2$ and $\ell \in [k/2,2k/3]$, while the fourth one follows since $k\geq 10$. We conclude that $\Pr[|P\cap S_i|=\ell]\geq \frac{1}{\sqrt{k}}\left(\frac{k}{\ell}-1\right)^{2\ell-k}$, as desired.
\end{proof}

\section{Missing Proofs of \Cref{sec:FGC}}

\exantefgc*
\begin{proof}

   Let $P$ be the panel selected by FGC, $c' = c(N)$, and $a \in C$ be any candidate. Following the exact same arguments as in the proof of  \Cref{main_theorem}, for getting  distortion of at most $(1 + \varepsilon)$, we will show that $\Pr[c(P) = a] \le \frac{\varepsilon \soc(c')}{2m (\soc(a) - \soc(c'))}.$
    
    When Fair Greedy Capture breaks ties in a consistent way (with respect to the agents that are located at the boundaries of a ball), it  can be equivalently described as following. The voters are partitioned into $k'\leq k$ sets, $S_1,\ldots S_{k'}$, where each $S_j$ with $j<k'$ contains $\lceil n/k \rceil$ voters and the last set contains at most $\lceil n/k \rceil$ voters. Then, the sampling procedure is executed into two stages. In Stage 1, from each $S_j$, we sample one voter with probability $1/\lceil n/k \rceil$ and in Stage 2, the remaining probability is allocated equally among the voters that were not selected in Stage 1. Note, that because $|S_{k'}|\leq \lceil  n/k \rceil$, it is probable no agent to be selected from $S_{k'}$, and thus we denote with $K\in [k'-1,k']$ the number of  voters that are selected  in Stage 1. Then, if $p_1,\ldots p_k$ are the random variables that indicate the agents in the panel $P$,  each $p_j$ with $j\leq K$ is sampled from $S_j$ independently  and the last $k-K$ agents are sampled from $N\setminus \{p_1,\ldots, p_{K}\}$ uniformly at random without replacement. Let $X_i = d(p_i, a) - d(p_i,c(N))$. A panel $P$ returns the candidate $a$ instead of $c(N)$, when $X=\sum_{i=1}^k X_i\leq 0$. Let $X'=\sum_{i=1}^K X_i$. 
    Next, we will show that (a) $\Pr[X'\leq kcD/4 ]\leq \frac{\varepsilon \soc(c')}{4m (\soc(a) - \soc(c'))}$ and  (b) $\Pr[X\leq 0 \mid X'> kcD/4]\leq \frac{\varepsilon \soc(c')}{4m (\soc(a) - \soc(c'))}$. Then, we have that 
    \begin{align*}
        \Pr[X\leq 0] &= \Pr[X\leq 0 \mid X' \leq kcD/4]\cdot \Pr[X'\leq kcD/4] + \Pr[X\leq 0 \mid X' > kcD/4]\cdot \Pr[X'> kcD/4] \\
        &\leq \frac{\varepsilon \soc(c')}{2m (\soc(a) - \soc(c'))}
    \end{align*}
    where the last inequality follows from (a) and (b), and the theorem follows. 

    We start by showing (a). Using identical arguments as in Case 1 in the proof of~\Cref{main_theorem}, we conclude that we can choose a  $1 \ge c \ge \frac{\varepsilon}{4} $ such that $\soc(a) - \soc(c(N)) = cnD$. 
    Thus, we have
    \begin{align*}
       \E[X']=\sum_{j=1}^{k'} \sum_{i\in S_j} \Pr[ i \text{ is selected}] \cdot \left( d(i,a)-d(i,c(N)) \right) 
        \geq&  \sum_{j=1}^{k'} \sum_{i\in S_j} \frac{1}{\lceil n/k\rceil}\cdot \left( d(i,a)-d(i,c(N)) \right) \\
        \geq&  \frac{k}{2n}\cdot cnD= \frac{kcD}{2}  
    \end{align*}
    Since the random variables $X_i$ for $i = 1, \dots, K$ are independent, we can apply Hoeffding's inequality and obtain 
     \begin{align*}
        &\Pr\left[X' <  \frac{kcD}{4}\right]\\ 
        &\le \Pr\left[\mid X' - \E[X'] \mid \ge  \frac{kcD}{4}\right] \\
        &\le 2 \exp\left(-2 \left( \frac{kcD}{4}\right)^2/ (K4D^2)\right) \\
        &\le 2 \exp\left(-kc^2/32\right).
    \end{align*}
    where the last inequality follows since $K\leq k$. 
    Thus, it is sufficient to choose $$k \ge \frac{\ln(4\frac{\soc(a) -  \soc(c')}{\soc(c'))}m/\varepsilon)64}{c^2}.$$ Now, since from~\Cref{thm:ex-post-FGC}, we know that $\frac{\soc(a) -  \soc(c')}{\soc(c'))}$ is at least equal to  $\varepsilon/2$  and at most $127$, we conclude that there is a    
    $k \in \mathcal{O}(\frac{\ln(m/\varepsilon)}{\varepsilon^2})$, for which (a) is true.

Next,  we show (b). First, we note that we can assume that $k - K \ge \frac{kc}{4}$ since otherwise $X$ cannot be negative anymore. We have that 
    \begin{align*}
        \E[X - X']= 
    \frac{k-K}{n - K} \cdot (cnD-X').
    \end{align*}
Using this, we get that 
\begin{align*}
    &\Pr[X < 0 \mid X' > kcD/4] = \Pr[X - X' < -X' \mid X' > kcD/4] \\ &= \Pr[X - X' - \E[X-X'] < -X' - \frac{k-K}{n - K} \cdot (cnD-X') \mid X' > kcD/4] \\
    &\le \Pr[\lvert X - X' - \E[X-X']\rvert  > \lvert -X' - \frac{k-K}{n - K} \cdot (cnD-X')\rvert \mid X' > kcD/4] \\
    &= \Pr[\lvert X - X' - \E[X-X']\rvert  > \lvert X' (1 - \frac{k-K}{n - K}) +  \frac{k-K}{n - K}(cnD)\rvert \mid X' > kcD/4] \\
    &\le  \Pr[\lvert X - X' - \E[X-X']\rvert  > \frac{kcD}{8} \mid X' > kcD/4]
\end{align*}
where the last inequality follows from the fact that $k\leq 2n$.

Thus, we can also apply Serfling's inequality (\Cref{thm:serfling74}) here and obtain that 
\begin{align*}
    &\Pr[X < 0 \mid X' > kcD/4] \\
&\le 2\exp(- \frac{2(kcD)^2}{64(k-K)D^2}) \\
&\le 2\exp(- \frac{2k^2c^2}{64 \frac{kc}{4}}) \\
& = 2\exp(- \frac{kc^3}{8}).
\end{align*}
Thus, to get a probability of at most $\frac{\varepsilon \soc(c')}{4m (\soc(a) - \soc(c'))}$ of selecting $a$ in this case, we choose $k \ge\frac{8}{c^3} \ln(4\frac{\soc(a) - \soc(c')}{\soc(c')}m/\varepsilon)$. As before since from~\Cref{thm:ex-post-FGC}, we know that $ \frac{\soc(a) - \soc(c')}{\soc(c')}$ is at least equal to  $\varepsilon/2$  and at most equal to a constant,  this gives us the desired asymptotic bound. 
\end{proof}

\end{document}

%% file: figures/at.tex
\begin{tikzpicture}[scale = 0.75]

\definecolor{darkgray176}{RGB}{176,176,176}
\definecolor{darkorange25512714}{RGB}{255,127,14}
\definecolor{lightgray204}{RGB}{204,204,204}
\definecolor{orange}{RGB}{255,165,0}
\definecolor{steelblue31119180}{RGB}{31,119,180}

\begin{axis}[
legend cell align={left},
legend style={fill opacity=0.8, draw opacity=1, text opacity=1, draw=lightgray204},
tick align=outside,
tick pos=left,
title={a) ESS Austria},
x grid style={darkgray176},
xlabel={Panel Size},
xmin=-1.4, xmax=41.4,
xtick style={color=black},
y grid style={darkgray176},
ylabel={Distortion},
ymin=0.985930834007133, ymax=1.32850522005963,
ytick style={color=black}
]
\path [draw=blue, fill=blue, opacity=0.15]
(axis cs:1,1.28783419982522)
--(axis cs:1,1.31131539417918)
--(axis cs:2,1.31293365705724)
--(axis cs:3,1.24655154411007)
--(axis cs:4,1.18173765353697)
--(axis cs:5,1.1416558923108)
--(axis cs:6,1.09371395933244)
--(axis cs:7,1.0761405894599)
--(axis cs:8,1.06173567911261)
--(axis cs:9,1.04865760454883)
--(axis cs:10,1.04327782085851)
--(axis cs:11,1.03415334611072)
--(axis cs:12,1.03106124594467)
--(axis cs:13,1.03263797791403)
--(axis cs:14,1.02641437476251)
--(axis cs:15,1.02583044715865)
--(axis cs:16,1.02285185252099)
--(axis cs:17,1.01819795721339)
--(axis cs:18,1.01795198287849)
--(axis cs:19,1.01760214033389)
--(axis cs:20,1.01628466902289)
--(axis cs:21,1.01380316905107)
--(axis cs:22,1.01320051395914)
--(axis cs:23,1.0137263603832)
--(axis cs:24,1.01316950865734)
--(axis cs:25,1.01084707800286)
--(axis cs:26,1.01104214195654)
--(axis cs:27,1.01161457102236)
--(axis cs:28,1.01022705317459)
--(axis cs:29,1.01026005273353)
--(axis cs:30,1.01063848470392)
--(axis cs:31,1.00907825245733)
--(axis cs:32,1.00804039787255)
--(axis cs:33,1.00744158439492)
--(axis cs:34,1.00816106173571)
--(axis cs:35,1.0071749125053)
--(axis cs:36,1.00713474374014)
--(axis cs:37,1.00649953277532)
--(axis cs:38,1.00634142475975)
--(axis cs:39,1.00733899680323)
--(axis cs:40,1.0065277757964)
--(axis cs:40,1.0042355215244)
--(axis cs:40,1.0042355215244)
--(axis cs:39,1.00504036820517)
--(axis cs:38,1.00404288804185)
--(axis cs:37,1.00347218858748)
--(axis cs:36,1.00447510010986)
--(axis cs:35,1.0045708424339)
--(axis cs:34,1.00537924118829)
--(axis cs:33,1.00456443637388)
--(axis cs:32,1.00509940303145)
--(axis cs:31,1.00588944037467)
--(axis cs:30,1.00735581494888)
--(axis cs:29,1.00678225766647)
--(axis cs:28,1.00704346518061)
--(axis cs:27,1.00810431887724)
--(axis cs:26,1.00751618651106)
--(axis cs:25,1.00745428764514)
--(axis cs:24,1.00989143016946)
--(axis cs:23,1.0099216669184)
--(axis cs:22,1.00941801937966)
--(axis cs:21,1.00981892213654)
--(axis cs:20,1.01229287800551)
--(axis cs:19,1.01282197613971)
--(axis cs:18,1.01336624426271)
--(axis cs:17,1.01342701851342)
--(axis cs:16,1.01773081404541)
--(axis cs:15,1.02041143763575)
--(axis cs:14,1.02064783971469)
--(axis cs:13,1.02674430576437)
--(axis cs:12,1.02442655827373)
--(axis cs:11,1.02749140482088)
--(axis cs:10,1.03495299728949)
--(axis cs:9,1.04023550597317)
--(axis cs:8,1.05155100777739)
--(axis cs:7,1.0642579590521)
--(axis cs:6,1.07876894847916)
--(axis cs:5,1.1251799596432)
--(axis cs:4,1.16355337237624)
--(axis cs:3,1.22749442548552)
--(axis cs:2,1.28796805295996)
--(axis cs:1,1.28783419982522)
--cycle;

\path [draw=orange, fill=orange, opacity=0.3]
(axis cs:1,1.28702299626911)
--(axis cs:1,1.31084585664729)
--(axis cs:2,1.30945444767946)
--(axis cs:3,1.2182469230074)
--(axis cs:4,1.17673555041767)
--(axis cs:5,1.11846416009636)
--(axis cs:6,1.07398947722634)
--(axis cs:7,1.05760836153482)
--(axis cs:8,1.04510773430989)
--(axis cs:9,1.03738354618484)
--(axis cs:10,1.02889445063884)
--(axis cs:11,1.02211301133895)
--(axis cs:12,1.01856829282434)
--(axis cs:13,1.01663935193287)
--(axis cs:14,1.01517244829502)
--(axis cs:15,1.01446323883157)
--(axis cs:16,1.01210232865545)
--(axis cs:17,1.00944045870358)
--(axis cs:18,1.00983147661086)
--(axis cs:19,1.0087303381383)
--(axis cs:20,1.00739143067437)
--(axis cs:21,1.00729765227112)
--(axis cs:22,1.00606954927345)
--(axis cs:23,1.00669558681949)
--(axis cs:24,1.00580852993984)
--(axis cs:25,1.005407985472)
--(axis cs:26,1.00502279737218)
--(axis cs:27,1.00476009350503)
--(axis cs:28,1.00479081372098)
--(axis cs:29,1.00453181306956)
--(axis cs:30,1.00422692918008)
--(axis cs:31,1.00399956115083)
--(axis cs:32,1.00326976480078)
--(axis cs:33,1.00315189485684)
--(axis cs:34,1.00274621375959)
--(axis cs:35,1.00348713556216)
--(axis cs:36,1.00301325004479)
--(axis cs:37,1.0030520367898)
--(axis cs:38,1.00271463335962)
--(axis cs:39,1.00290554099789)
--(axis cs:40,1.00275964800842)
--(axis cs:40,1.00159520105118)
--(axis cs:40,1.00159520105118)
--(axis cs:39,1.00161735822691)
--(axis cs:38,1.00167178429918)
--(axis cs:37,1.0017969691554)
--(axis cs:36,1.00196802972001)
--(axis cs:35,1.00185450222464)
--(axis cs:34,1.00160543619241)
--(axis cs:33,1.00178315623236)
--(axis cs:32,1.00187163929962)
--(axis cs:31,1.00229107977116)
--(axis cs:30,1.00274971785872)
--(axis cs:29,1.00283439701164)
--(axis cs:28,1.00282342875822)
--(axis cs:27,1.00274443528257)
--(axis cs:26,1.00296050705742)
--(axis cs:25,1.0034580541148)
--(axis cs:24,1.00349426456216)
--(axis cs:23,1.00419033159291)
--(axis cs:22,1.00369877806655)
--(axis cs:21,1.00467679221568)
--(axis cs:20,1.00450097773483)
--(axis cs:19,1.0056147487833)
--(axis cs:18,1.00677803551554)
--(axis cs:17,1.00611852222402)
--(axis cs:16,1.00816257018695)
--(axis cs:15,1.01043709623883)
--(axis cs:14,1.01136982707538)
--(axis cs:13,1.01237508121753)
--(axis cs:12,1.01350790261806)
--(axis cs:11,1.01695467602945)
--(axis cs:10,1.02252629683796)
--(axis cs:9,1.03046966317636)
--(axis cs:8,1.03649158441491)
--(axis cs:7,1.04780039318598)
--(axis cs:6,1.06225486157086)
--(axis cs:5,1.10436769042724)
--(axis cs:4,1.16109817886673)
--(axis cs:3,1.2032909317666)
--(axis cs:2,1.28564467298974)
--(axis cs:1,1.28702299626911)
--cycle;

\addplot [semithick, steelblue31119180]
table {%
1 1.2995747970022
2 1.3004508550086
3 1.2370229847978
4 1.1726455129566
5 1.133417925977
6 1.0862414539058
7 1.070199274256
8 1.056643343445
9 1.044446555261
10 1.039115409074
11 1.0308223754658
12 1.0277439021092
13 1.0296911418392
14 1.0235311072386
15 1.0231209423972
16 1.0202913332832
17 1.0158124878634
18 1.0156591135706
19 1.0152120582368
20 1.0142887735142
21 1.0118110455938
22 1.0113092666694
23 1.0118240136508
24 1.0115304694134
25 1.009150682824
26 1.0092791642338
27 1.0098594449498
28 1.0086352591776
29 1.0085211552
30 1.0089971498264
31 1.007483846416
32 1.006569900452
33 1.0060030103844
34 1.006770151462
35 1.0058728774696
36 1.005804921925
37 1.0049858606814
38 1.0051921564008
39 1.0061896825042
40 1.0053816486604
};
\addlegendentry{Uniform}
\addplot [semithick, darkorange25512714]
table {%
1 1.2989344264582
2 1.2975495603346
3 1.210768927387
4 1.1689168646422
5 1.1114159252618
6 1.0681221693986
7 1.0527043773604
8 1.0407996593624
9 1.0339266046806
10 1.0257103737384
11 1.0195338436842
12 1.0160380977212
13 1.0145072165752
14 1.0132711376852
15 1.0124501675352
16 1.0101324494212
17 1.0077794904638
18 1.0083047560632
19 1.0071725434608
20 1.0059462042046
21 1.0059872222434
22 1.00488416367
23 1.0054429592062
24 1.004651397251
25 1.0044330197934
26 1.0039916522148
27 1.0037522643938
28 1.0038071212396
29 1.0036831050406
30 1.0034883235194
31 1.003145320461
32 1.0025707020502
33 1.0024675255446
34 1.002175824976
35 1.0026708188934
36 1.0024906398824
37 1.0024245029726
38 1.0021932088294
39 1.0022614496124
40 1.0021774245298
};
\addlegendentry{FGC}
\addplot [semithick, red, opacity=0.1, dashed, forget plot]
table {%
-0.95 1
41.95 1
};
\end{axis}

\end{tikzpicture}

%% file: figures/adult.tex
\begin{tikzpicture}[scale = 0.75]

\definecolor{darkgray176}{RGB}{176,176,176}
\definecolor{darkorange25512714}{RGB}{255,127,14}
\definecolor{lightgray204}{RGB}{204,204,204}
\definecolor{orange}{RGB}{255,165,0}
\definecolor{steelblue31119180}{RGB}{31,119,180}

\begin{axis}[
legend cell align={left},
legend style={fill opacity=0.8, draw opacity=1, text opacity=1, draw=lightgray204},
tick align=outside,
tick pos=left,
title={b) Adult},
x grid style={darkgray176},
xlabel={Panel Size},
xmin=-1.4, xmax=41.4,
xtick style={color=black},
y grid style={darkgray176},
ylabel={Distortion},
ymin=0.996596234850859, ymax=1.07890467287251,
ytick style={color=black}
]
\path [draw=blue, fill=blue, opacity=0.15]
(axis cs:1,1.05543273086625)
--(axis cs:1,1.06543666260055)
--(axis cs:2,1.07321867610959)
--(axis cs:3,1.03308721982606)
--(axis cs:4,1.02869929287312)
--(axis cs:5,1.01700659844303)
--(axis cs:6,1.01595008166947)
--(axis cs:7,1.01154429337105)
--(axis cs:8,1.01029540792857)
--(axis cs:9,1.00869745437122)
--(axis cs:10,1.00686439640338)
--(axis cs:11,1.00686476704018)
--(axis cs:12,1.0068873479431)
--(axis cs:13,1.00585157136645)
--(axis cs:14,1.00600067486182)
--(axis cs:15,1.00516391833365)
--(axis cs:16,1.00455016213375)
--(axis cs:17,1.00424185672541)
--(axis cs:18,1.00418685397856)
--(axis cs:19,1.00357863282499)
--(axis cs:20,1.00334770676894)
--(axis cs:21,1.0031228158674)
--(axis cs:22,1.00357290738601)
--(axis cs:23,1.00283272311303)
--(axis cs:24,1.00302224161288)
--(axis cs:25,1.00296964729653)
--(axis cs:26,1.00272444494729)
--(axis cs:27,1.00274793313915)
--(axis cs:28,1.00255711306442)
--(axis cs:29,1.00238279134046)
--(axis cs:30,1.00228469174575)
--(axis cs:31,1.0022399383241)
--(axis cs:32,1.00223595495493)
--(axis cs:33,1.00216096355449)
--(axis cs:34,1.00225806045856)
--(axis cs:35,1.00223829789934)
--(axis cs:36,1.00206381534062)
--(axis cs:37,1.00190975992313)
--(axis cs:38,1.00193631937817)
--(axis cs:39,1.00173008712367)
--(axis cs:40,1.0019257503362)
--(axis cs:40,1.0014250204018)
--(axis cs:40,1.0014250204018)
--(axis cs:39,1.00114339897353)
--(axis cs:38,1.00138068577423)
--(axis cs:37,1.00133943428527)
--(axis cs:36,1.00137531824498)
--(axis cs:35,1.00148582983266)
--(axis cs:34,1.00155784451183)
--(axis cs:33,1.00158419908791)
--(axis cs:32,1.00148178151467)
--(axis cs:31,1.0015441346527)
--(axis cs:30,1.00150562381985)
--(axis cs:29,1.00152062179514)
--(axis cs:28,1.00172516738758)
--(axis cs:27,1.00190176925525)
--(axis cs:26,1.00182048275871)
--(axis cs:25,1.00189495285427)
--(axis cs:24,1.00199815340033)
--(axis cs:23,1.00176315521737)
--(axis cs:22,1.00225258865439)
--(axis cs:21,1.0022121876326)
--(axis cs:20,1.00232871972906)
--(axis cs:19,1.00238949668261)
--(axis cs:18,1.00286329270224)
--(axis cs:17,1.00290794357979)
--(axis cs:16,1.00329904504985)
--(axis cs:15,1.00357318680794)
--(axis cs:14,1.00398422874258)
--(axis cs:13,1.00405855758875)
--(axis cs:12,1.0048855622233)
--(axis cs:11,1.00485085826822)
--(axis cs:10,1.00486844579262)
--(axis cs:9,1.00600349907558)
--(axis cs:8,1.00752266072423)
--(axis cs:7,1.00882036243815)
--(axis cs:6,1.01243055396813)
--(axis cs:5,1.01325090473457)
--(axis cs:4,1.02242637935208)
--(axis cs:3,1.02609171221754)
--(axis cs:2,1.06362050909521)
--(axis cs:1,1.05543273086625)
--cycle;

\path [draw=orange, fill=orange, opacity=0.3]
(axis cs:1,1.06521240171563)
--(axis cs:1,1.07516338023517)
--(axis cs:2,1.06681603470756)
--(axis cs:3,1.02590870998739)
--(axis cs:4,1.01678175192453)
--(axis cs:5,1.0088940701874)
--(axis cs:6,1.00896731818416)
--(axis cs:7,1.00461605924825)
--(axis cs:8,1.00394800823001)
--(axis cs:9,1.00349988535287)
--(axis cs:10,1.00330049329963)
--(axis cs:11,1.00270065005403)
--(axis cs:12,1.00215191028456)
--(axis cs:13,1.00194314873089)
--(axis cs:14,1.00185835778424)
--(axis cs:15,1.00135564796798)
--(axis cs:16,1.00147828802037)
--(axis cs:17,1.00136227853288)
--(axis cs:18,1.00142467539044)
--(axis cs:19,1.00105788378765)
--(axis cs:20,1.00123361450885)
--(axis cs:21,1.00109969304547)
--(axis cs:22,1.00098919966903)
--(axis cs:23,1.0010594072979)
--(axis cs:24,1.00096204937287)
--(axis cs:25,1.00078986210645)
--(axis cs:26,1.00094059590561)
--(axis cs:27,1.0007567923668)
--(axis cs:28,1.00094691757433)
--(axis cs:29,1.00081879145639)
--(axis cs:30,1.00076338698044)
--(axis cs:31,1.00068652448194)
--(axis cs:32,1.00083320988873)
--(axis cs:33,1.0007947069998)
--(axis cs:34,1.0007955593996)
--(axis cs:35,1.00075793348057)
--(axis cs:36,1.00069815943524)
--(axis cs:37,1.00056517476554)
--(axis cs:38,1.00073813739461)
--(axis cs:39,1.00060923199248)
--(axis cs:40,1.00066235909029)
--(axis cs:40,1.00044128290011)
--(axis cs:40,1.00044128290011)
--(axis cs:39,1.00038952081472)
--(axis cs:38,1.00049313883099)
--(axis cs:37,1.00035656790406)
--(axis cs:36,1.00047447293236)
--(axis cs:35,1.00054228681943)
--(axis cs:34,1.0005673331788)
--(axis cs:33,1.0005530859362)
--(axis cs:32,1.00057371155767)
--(axis cs:31,1.00046764682446)
--(axis cs:30,1.00051818778636)
--(axis cs:29,1.00059837120721)
--(axis cs:28,1.00070308662767)
--(axis cs:27,1.0005144286292)
--(axis cs:26,1.00063566781079)
--(axis cs:25,1.00051136661595)
--(axis cs:24,1.00069002637793)
--(axis cs:23,1.0007796535769)
--(axis cs:22,1.00073828274817)
--(axis cs:21,1.00082337897093)
--(axis cs:20,1.00095569438515)
--(axis cs:19,1.00078355751675)
--(axis cs:18,1.00108988985356)
--(axis cs:17,1.00101783290552)
--(axis cs:16,1.00110555173363)
--(axis cs:15,1.00097878140322)
--(axis cs:14,1.00142505357376)
--(axis cs:13,1.00148473688951)
--(axis cs:12,1.00158400384824)
--(axis cs:11,1.00190100899117)
--(axis cs:10,1.00230326974237)
--(axis cs:9,1.00256581647153)
--(axis cs:8,1.00292988624959)
--(axis cs:7,1.00321779937295)
--(axis cs:6,1.00656745138944)
--(axis cs:5,1.006345262109)
--(axis cs:4,1.01303778628467)
--(axis cs:3,1.02103677953301)
--(axis cs:2,1.05823132080404)
--(axis cs:1,1.06521240171563)
--cycle;

\addplot [semithick, steelblue31119180]
table {%
1 1.0604346967334
2 1.0684195926024
3 1.0295894660218
4 1.0255628361126
5 1.0151287515888
6 1.0141903178188
7 1.0101823279046
8 1.0089090343264
9 1.0073504767234
10 1.005866421098
11 1.0058578126542
12 1.0058864550832
13 1.0049550644776
14 1.0049924518022
15 1.0043685525708
16 1.0039246035918
17 1.0035749001526
18 1.0035250733404
19 1.0029840647538
20 1.002838213249
21 1.00266750175
22 1.0029127480202
23 1.0022979391652
24 1.0025101975066
25 1.0024323000754
26 1.002272463853
27 1.0023248511972
28 1.002141140226
29 1.0019517065678
30 1.0018951577828
31 1.0018920364884
32 1.0018588682348
33 1.0018725813212
34 1.0019079524852
35 1.001862063866
36 1.0017195667928
37 1.0016245971042
38 1.0016585025762
39 1.0014367430486
40 1.001675385369
};
\addlegendentry{Uniform}
\addplot [semithick, darkorange25512714]
table {%
1 1.0701878909754
2 1.0625236777558
3 1.0234727447602
4 1.0149097691046
5 1.0076196661482
6 1.0077673847868
7 1.0039169293106
8 1.0034389472398
9 1.0030328509122
10 1.002801881521
11 1.0023008295226
12 1.0018679570664
13 1.0017139428102
14 1.001641705679
15 1.0011672146856
16 1.001291919877
17 1.0011900557192
18 1.001257282622
19 1.0009207206522
20 1.001094654447
21 1.0009615360082
22 1.0008637412086
23 1.0009195304374
24 1.0008260378754
25 1.0006506143612
26 1.0007881318582
27 1.000635610498
28 1.000825002101
29 1.0007085813318
30 1.0006407873834
31 1.0005770856532
32 1.0007034607232
33 1.000673896468
34 1.0006814462892
35 1.00065011015
36 1.0005863161838
37 1.0004608713348
38 1.0006156381128
39 1.0004993764036
40 1.0005518209952
};
\addlegendentry{FGC}
\addplot [semithick, red, opacity=0.1, dashed, forget plot]
table {%
-0.95 1
41.95 1
};
\end{axis}

\end{tikzpicture}